\documentclass[10pt,twoside]{article}

\usepackage[a4paper,textwidth=5.77in,textheight=7.63in,centering,headheight=10pt,headsep=14pt]{geometry}

\usepackage{mathpazo}
\usepackage{microtype}
\usepackage{fancyhdr}
\usepackage{titlesec}
\titleformat{\section}[block]{\normalfont\normalsize\centering\scshape}{\thesection.}{0.55em}{}
\titlespacing*{\section}{0pt}{2.1\baselineskip}{0.85\baselineskip}
\titleformat{\subsection}[block]{\normalfont\normalsize\bfseries}{\thesubsection}{0.65em}{}
\titlespacing*{\subsection}{0pt}{1.6\baselineskip}{0.55\baselineskip}
\titleformat{\subsubsection}[runin]{\normalfont\normalsize\bfseries}{\thesubsubsection}{0.6em}{}[.]
\titlespacing*{\subsubsection}{0pt}{1.25\baselineskip}{0.5em}

\usepackage{amsmath,amssymb,amsthm,mathtools}
\usepackage{enumitem}
\usepackage[hidelinks,hypertexnames=false]{hyperref}
\usepackage{booktabs}
\usepackage{array}
\usepackage{tabularx}
\usepackage{float}

\theoremstyle{plain}
\newtheorem{theorem}{Theorem}
\newtheorem{proposition}{Proposition}
\newtheorem{lemma}{Lemma}
\newtheorem{corollary}{Corollary}

\newtheorem{observation}{Observation}

\hypersetup{
  pdftitle={Singleton-Attainability and Transparent Access in Matching},
  pdfauthor={Szilvia P\'apai}
}

\theoremstyle{definition}

\newtheorem{example}{Example}

\theoremstyle{remark}

\newcommand{\lift}[1]{^{\uparrow #1}}
\newcommand{\trunc}[1]{^{\le #1}}
\newcommand{\defhead}[1]{\par\bigskip\noindent\textbf{#1.}\ }
\newcommand{\asgn}{\mathord{\text{-}}}
\newcommand{\rk}{\operatorname{rk}}

\makeatletter
\renewcommand{\@biblabel}[1]{}
\makeatother

\title{Singleton-Attainability \\and Transparent Access in Matching\thanks{The author gratefully acknowledges research funding from the Natural Sciences and Engineering Research Council of Canada (NSERC) through Discovery Grant RGPIN-2026-05471. The author also gratefully acknowledges support from the Hausdorff Research Institute for Mathematics (HIM) during the Trimester Program ``Advances in Mechanism Design,'' funded by the Deutsche Forschungsgemeinschaft (DFG, German Research Foundation) under Germany's Excellence Strategy -- EXC-2047/2 -- 390685813.}}
\author{Szilvia P\'apai\thanks{Concordia University and CIREQ; e-mail: \href{mailto:szilvia.papai@concordia.ca}{szilvia.papai@concordia.ca}}}

\date{September 22, 2026}

\begin{document}
\maketitle
\thispagestyle{empty}

\begin{abstract}

Matching mechanisms differ in how much of an agent's preference ranking must be determined and reported to obtain a particular object. A mechanism is \textit{singleton-attainable} (SA) if every object that an agent can obtain through some report can also be obtained by reporting only that object as acceptable. With an SA mechanism, once an attainable object has been identified, the agent need not rank or report any other object. Singleton-attainability identifies a distinct dimension in matching theory and market design: transparent access to attainable outcomes, separate from incentives, stability, welfare or equity.

We establish general conditions for SA and derive its strategic implications. Top-lift invariance and truncation invariance together imply SA, while strategyproofness and stability each imply SA. By contrast, no Pareto improvement over a
strategyproof, individually rational, and non-wasteful mechanism is SA. In particular, every Pareto improvement over Deferred Acceptance violates SA. We introduce report width, which measures how many acceptable objects may have to be reported to obtain an object. SA mechanisms have report width one.

Report width is unbounded for a large class of efficient mechanisms that Pareto-improve Deferred Acceptance. Stable selection with report-induced priorities has width one when priorities are monotone and maximal width under reverse priority dominance. Rank-welfare maximization has width one
when the outside-option rank is fixed and maximal width when it is report-dependent. 
These results reveal a structural divide, which we call the
\textit{width dichotomy}: across all mechanisms and families in our
classification and all structural classes we study, report width is
either one or unbounded.

\end{abstract}

\section{Introduction}\label{sec:intro}

Suppose an agent can obtain an object by submitting some preference list. Can she obtain the same object by reporting only this object as acceptable? For Deferred Acceptance (DA),\footnote{The Agent-Proposing DA mechanism is referred to simply as DA throughout the paper. The version in which objects propose is referred to as the Object-Proposing DA.} Immediate Acceptance (IA), Parallel mechanisms and many others, the answer is yes. For Efficiency-Adjusted Deferred Acceptance (EADA) and several equity or rank-based mechanisms,  the answer is no in some markets.\footnote{Appendix~\ref{app:mechanisms} collects definitions and references for the mechanisms central to the analysis.} If EADA is used in school choice, for example, a student may need to list an additional school below the target school to trigger a rejection chain that eventually frees a seat at her target school. Omitting this lower-ranked school can make the target school unattainable. 
This is not merely hypothetical. A restricted version of EADA has recently been introduced as an option in secondary-school admissions in Flanders, where the design studies noted this specific effect.\footnote{Section~\ref{sec:practice}
discusses the Flemish implementation of EADA and connects this effect
to the broader phenomenon of indirect access studied in this paper.} 
For each of the mechanisms in this second group, the number of objects that may have to be listed to obtain a desired object grows without bound as the market grows. 
This reveals a broader pattern: every primary mechanism and family in
our classification has report width one or unbounded width, and the
same is true of every structural mechanism class whose boundedness we
determine, with nothing in between. The pattern is not a logical necessity: examples with bounded width greater than one are easy to construct. It is a substantive feature of the mechanisms that arise in theory and practice. 
We call it the \emph{width dichotomy}.

The paper studies this accessibility question in the standard many-to-one matching model with preferences on one side and priorities on the other, known as the school choice model. The two sides may represent students and schools, applicants and university programs, workers and positions, or other participants and institutions with the same structure. We use ``agent'' and ``object'' as generic terms. An object may have multiple copies, represented by its capacity, and all copies share the object's priority ordering. Throughout, ``object'' refers to an object type rather than an individual copy, and report width counts objects. 

In many such systems, agents submit ordered lists of acceptable objects. Constructing a long list requires more than entering additional objects: it requires learning about them, comparing them, and deciding which are acceptable. These tasks may be demanding for boundedly rational agents and may depend on access to information or advice. List-length constraints limit how many objects can be reported, but may also prevent truthful reporting and force agents to decide which objects to omit. Unrestricted and constrained reporting can therefore be informationally demanding in different ways. Hakimov and Khanna (2025) find substantial errors even when subjects are rewarded simply for accurately reporting their induced ordinal rankings, with errors increasing when the ranking requirement is more complex.

Fix an agent $i$ and the reports $R_{-i}$ of the other agents. An object $c$ is attainable for $i$ at $R_{-i}$ if some report by $i$ assigns $c$ to her. We call a mechanism \emph{singleton-attainable} (SA) if, for every $i$, $R_{-i}$, and attainable object $c$, the singleton report $(c)$ also assigns $c$ to $i$. The minimum report size of $c$ at $R_{-i}$ is the length of the shortest report by $i$ that attains $c$. For a fixed number of objects in the market, report width is the largest minimum positive report size that can arise across admissible markets, agents, reports of the other agents, and attainable objects. For individually rational mechanisms, width one is exactly SA (Observation~\ref{obs:width-one}). Width is bounded if the worst-case minimum report size remains below a fixed bound as the number of objects increases, and unbounded otherwise. SA does not identify which objects are attainable or which attainable object the agent prefers. It concerns the structure of access once an attainable object has been identified. For a singleton-attainable mechanism, access to such an object is direct: no other acceptable object or comparison needs to appear in the report. For non-SA mechanisms, access may depend on additional entries that affect the course of the mechanism. In this sense, SA captures an internal form of transparency in the mechanism itself.

This distinction matters for list-length caps, which are common in
practice. New York City, for example, limited high-school applicants'
lists to twelve programs for many years (Roth, 2015). Such a limit
creates a strategic choice about which alternatives to include, since
omitting a lower-ranked safe option can leave the applicant unmatched.
This choice arises even under DA and is not captured by SA. SA concerns
a different effect of list-length restrictions: with an SA mechanism
and the other applicants' reports fixed, every school attainable with
unrestricted reporting remains attainable with any positive
list-length bound, since the singleton report listing that school
attains it. Without SA, a cap can make a school unattainable even to
an applicant who lists it, because the additional schools needed to
attain it cannot all be included.

Singleton-attainability has an immediate consequence for what the mechanism reveals through singleton reports. 
Gonczarowski, Heffetz, and Thomas (2025) recover an agent's
option set, her menu, from a strategyproof mechanism by testing
each singleton report and including an object if its singleton
report returns it. Singleton-attainability is exactly the condition
that makes this recovery valid without strategyproofness: for
every agent and every profile of the other agents' reports, her
option set consists exactly of the objects that she obtains from
their respective singleton reports (Observation~\ref{obs:oppset}).
SA thus identifies the precise feature which makes the option set recoverable from singleton tests, whether or not truthful reporting is optimal. While implied by strategyproofness, SA itself is not an incentive property. It is a property of the report needed to obtain an object and is meaningful for both strategyproof and manipulable mechanisms.

Small report width, and specifically singleton-attainability, has both advantages and potential drawbacks. Smaller width reduces the informational requirements for obtaining an object, but it may also make a profitable deviation easier to identify if the mechanism is manipulable, even though report width measures the length of a successful report, not how precisely its objects must be ordered or how difficult such a report is to find.\footnote{A separate computational literature asks how difficult it is to find a profitable manipulation; see Klaus, Manlove, and Rossi (2016, Section 14.2.3) for a survey.} 
The existence of manipulable SA mechanisms is an informative feature of SA, without defining its scope. It is a transparency property which may make a profitable manipulation easier to identify, but it is distinct from the obvious manipulability notion of Troyan and Morrill (2020). Object-Proposing DA is SA and, though manipulable, it is not obviously manipulable. On the other hand, Rank-Equity\footnote{We use the name Rank-Equity for our extension of the mechanism of Ekici, Ertemel, and Yenmez (2026) to many-to-one matching with incomplete reports, which applies DA to report-induced priorities that favor agents who rank an object lower. Appendix~\ref{app:mechanisms} gives the definition.}  is not SA and is obviously manipulable for some tie-breakers, demonstrating that the payoff transparency feature of obvious manipulability is separate from the access transparency of SA.\footnote{Appendix~\ref{app:om} develops the comparison with obvious manipulation.} 

Singleton-attainability has several strategic implications for individually rational mechanisms. SA is equivalent to the existence of a singleton or empty best response in every best-response problem and to the sufficiency of such deviations for verifying Nash equilibrium (Proposition~\ref{prop:sa-best-response}). More generally, report width is the smallest bound such that, to verify Nash equilibrium, it suffices to check deviations listing no more than that many acceptable objects (Proposition~\ref{prop:best-response}). 
For individually rational mechanisms, SA is also equivalent to
the following strategic property: whenever an agent can
profitably manipulate, she can do so by reporting just one
truthfully acceptable object. Such a manipulation involves
withholding acceptable objects, without a preference reversal
or false acceptability claim
(Proposition~\ref{prop:singleton-manipulation}).

We refer to the mechanisms and mechanism families classified in
Table~\ref{tab:class} (Section~\ref{sec:classify}) collectively as
the \emph{primary mechanisms and families}. A closer look at the
non-SA primary mechanisms and families reveals two one-sided forms
of indirect access. EADA, DA-TTC, and several mechanism families with a report-dependent outside option rank satisfy \textit{top-lift invariance }(TLI) but violate \textit{truncation invariance} (TRI): an agent cannot lose an obtained object by moving it to the top of her report, but may lose it by deleting everything below it. We call this \textit{tail dependence}. Rank-Equity violates TLI but satisfies TRI: an agent may lose an obtained object by moving it to the top of her report, but cannot lose it by deleting everything below it. We call this \textit{head dependence}.
We show that TLI and TRI together imply SA (Proposition~\ref{prop:decomp}).
Every Preference Rank Partitioned (PRP) mechanism\footnote{The PRP family includes DA, IA, First-Preference-First, Parallel mechanisms, Secure Boston, and French Priority mechanisms, among others (Ayoade and P\'apai, 2023).} and every Taiwan Deduction mechanism,\footnote{The Taiwan Deduction mechanism has been used nationwide for high-school assignment in Taiwan since 2014. It uses DA with deductions from applicants' scores that weakly increase with the rank at which a school is listed (Dur et al., 2022).
This family includes DA, IA and Parallel mechanisms, among others.} as well as every rank-based primary mechanism with a fixed rank for the outside option, satisfies both invariances directly (Corollaries~\ref{cor:smp-applications} and~\ref{cor:fixed-rank-welfare-applications}). In fact, every primary mechanism and mechanism family on the SA side of the dichotomy satisfies both invariances. Strategyproofness implies both invariances and SA (Proposition~\ref{prop:sp-inv}), while stability provides a separate sufficient condition for SA (Proposition~\ref{prop:stable}).

The analysis also reveals a general tradeoff. No Pareto
improvement over a strategyproof, individually rational, and
non-wasteful mechanism is singleton-attainable, even without requiring
the improving mechanism to be Pareto efficient
(Theorem~\ref{thm:pareto-improvement-nonsa}). Within the present matching
environment, this strengthens the corresponding strategyproofness
frontier result of Alva and Manjunath (2019b): such a baseline
mechanism is Pareto-undominated not only within the class of
strategyproof mechanisms, but also within the larger class of SA
mechanisms. Since DA satisfies the three properties, every Pareto improvement
over DA violates SA (Corollary~\ref{cor:da-improvement-nonsa}).
Moreover, DA is Pareto efficient when all agents report at most one
acceptable object. This fits the familiar understanding of DA's
efficiency losses: richer reports may create priority claims on objects
ranked above an agent's eventual assignment, and these claims set off
the rejection cycles that produce the losses. Pareto improvements over
DA may recover welfare lost through such claims, but they cannot do so
while retaining SA.

We also introduce \emph{structural mechanism classes} that generalize and group the primary mechanisms and families and explain their access properties. An overview of this structural classification is presented in Table~\ref{tab:classes} (Section~\ref{sec:class-revisited}). The purpose of studying these large structural classes is to establish common sources of direct and indirect access. 
The analysis identifies three pairs of contrasting classes.
First, strategyproof mechanisms have width one, whereas
Positional DA-TC mechanisms, which are Pareto-efficient
mechanisms that Pareto-improve the strategyproof DA baseline,
have unbounded width (Theorem~\ref{thm:positional-width}).
Second, Stable Monotone Priority mechanisms have width one
(Theorem~\ref{thm:smp}), whereas Stable Reverse Priority
mechanisms have maximal width (Theorem~\ref{thm:srp-width}).
Third, Fixed Rank Welfare mechanisms have width one
(Theorem~\ref{thm:fixed-outside}), whereas Reported Rank Welfare
mechanisms have maximal width
(Theorem~\ref{thm:reported-outside}).
Together, these 
structural results encompass all the primary mechanisms and families in
our classification and identify the structures responsible for their
access properties.

Taken together, our results identify singleton-attainability as a dimension of mechanisms that has not been studied in its own right, distinct from the incentive, stability, welfare or equity properties through which matching mechanisms are usually studied and compared. Classifying familiar mechanisms by the transparency of the access they provide reorganizes them, placing mechanisms with similar incentive or welfare properties on opposite sides and mechanisms that are usually contrasted on the same side. This access dimension is structural: it is governed by identifiable features of a mechanism's procedure, culminating in the width dichotomy, and it has consequences for information, disclosure, and the effect of reporting constraints.

The paper is related to several notions of transparency and complexity. M\"oller (2026), Grigoryan and M\"oller (2025), and Hakimov and Raghavan (2026) study related questions of transparency, auditability, and verification, including whether an announced mechanism was followed and whether its outcome can be verified. Gonczarowski and Thomas (2024) quantify global structural complexities of matching mechanisms, including the effects of one agent's report on matchings and option sets and the information needed to represent or verify an outcome. Gonczarowski, Heffetz, and Thomas (2025) study descriptions that expose strategyproofness by revealing an agent's option set, or menu. Feng and Liu (2024) measure the partial information about other agents' preferences and schools' priorities needed to secure a student-school assignment, while Bonkoungou and Nesterov (2021) compare mechanisms through objects reachable by profitable strategies. We hold the mechanism and the reports of the other agents fixed and ask a different question: once an object has been identified as attainable, how much of the agent's own preference must she know and submit to reach it? Report width measures the support size of an access report, rather than the complexity of identifying the option set, describing the mechanism, or representing and verifying the resulting matching.

The practical motivation is also related to work on constrained lists and simplified preference reporting, including Haeringer and Klijn (2009), Arteaga et al.\ (2022), and Huang and Zhang (2025). Hatfield and Milgrom (2005, Theorem 10) obtain a singleton-access property for the agent-optimal stable mechanism in matching with contracts under substitutes, as a step toward strategyproofness. 
Zhang (2023) studies related invariance properties for random allocation mechanisms.
Here singleton access becomes the object of study across strategyproof and manipulable mechanisms, and report width quantifies the severity of its violation. 

The paper is organized as follows. Section~\ref{sec:model} introduces the model. Section~\ref{sec:sa} defines singleton-attainability and develops its immediate implications for option sets, manipulations, best responses, and equilibrium verification.
Section~\ref{sec:decomp} introduces the head-tail decomposition and relates
SA to strategyproofness and stability. Section~\ref{sec:classify} presents the primary classification. Sections~\ref{sec:causes-sa} and~\ref{sec:failure} identify structural causes of SA and of its violation respectively, with report width introduced in Section~\ref{sec:width}. 
Section~\ref{sec:class-revisited} brings these structural results together and establishes the width dichotomy. Section~\ref{sec:implications} develops the informational, disclosure, and practical implications, and Section~\ref{sec:conclusion} concludes. Appendix~\ref{app:mechanisms} collects definitions and references for the primary mechanisms central to the analysis, Appendix~\ref{app:proofs} contains omitted proofs, and Appendix~\ref{app:om} compares SA with obvious manipulation.

\section{Model}\label{sec:model}
We work with the standard unit-demand many-to-one matching model with priorities, known as the school choice model. There is a finite set of agents $I$ and a finite set $C$ of $m$ objects. Each object $c\in C$ has capacity $q_c\geq 1$, indicating $q_c$ identical copies with a common strict priority ordering over agents.

A \emph{matching} is a mapping $\mu:I\to C\cup\{\emptyset\}$ such that $|\mu(c)|\leq q_c$ for every $c\in C$, where $\mu_i=\mu(i)$, $\mu(c)=\{i\in I:\mu_i=c\}$, and $\emptyset$ denotes the outside option. We refer to $\mu_i$ as agent $i$'s \emph{assignment}, and to $\mu$ as the \emph{matching.} Let $\mathcal M$ denote the set of matchings. We write $i\asgn c$ to indicate that $\mu_i=c$.

The preference domain and the report space coincide. Each agent $i$ has a strict preference ordering $R_i$ over $C\cup\{\emptyset\}$.
Every object in the set $A_i(R_i)\subseteq C$ of acceptable objects is ranked strictly above $\emptyset$; every object outside $A_i(R_i)$ is below $\emptyset$ and considered unacceptable. The relative ordering of unacceptable objects is immaterial for our
purposes and is suppressed in the notation. We also refer to this as having \textit{incomplete reports.} 
We write $c\,R_i\,d$ when agent $i$ weakly prefers $c$ to $d$, and $c\,P_i\,d$ when she strictly prefers $c$ to $d$. Let $\mathcal R_i$ denote the set of all such preference orderings and let $\mathcal R =\prod_{i\in I}\mathcal R_i$ with $\mathcal R_{-i} =\prod_{j\neq i}\mathcal R_j$. 
Since only the acceptable objects and their order are displayed, we use
ordered-list notation: $R_i=(a,b,c)$ means
$a\,P_i\,b\,P_i\,c\,P_i\,\emptyset$, with every unlisted object ranked
below $\emptyset$. We write $(c)$ for the report that declares only $c$ acceptable and $\emptyset$ for the empty report.

Each object $c$ has a strict priority order $\succ_c$ over $I$, and $\succ \, =(\succ_c)_{c\in C}$ denotes the priority profile. When convenient, we write $\succ_c=(i,j,h,\ldots)$, which lists the agents starting with the highest priority. For all remaining agents, ranked below those listed, relative priority is left unspecified. 

A \textit{mechanism} assigns a matching to each report profile, given the
object capacity list $q=(q_c)_{c\in C}$ and, for mechanisms that use
priorities, the priority profile $\succ$. All mechanisms and mechanism
classes studied are defined for arbitrary object capacities. When a
mechanism was originally formulated only for unit capacities, we use the
many-to-one version defined in this paper. Since $q$ and, when relevant,
$\succ$ are fixed throughout unless stated otherwise, we suppress them from
the notation and write $\varphi:\mathcal R\to\mathcal M$, with
$\varphi_i(R)$ denoting agent $i$'s assignment.

 A matching $\mu$ is \emph{individually
rational}  at profile $R$ if $\mu_i\in A_i(R_i)\cup\{\emptyset\}$ for every agent $i$.  A mechanism is individually rational if it selects an individually
rational matching at every profile.
A matching $\mu$ is \emph{non-wasteful} at profile $R$ if, for every agent $i$ and object $c$, $c\,P_i\,\mu_i$ implies $|\mu(c)|=q_c$. A mechanism is non-wasteful if it selects a non-wasteful matching at every profile.

A mechanism $\varphi$ is \emph{strategyproof}  if, for every profile $R$, every agent $i$, and every alternative report $\widehat R_i$, $\varphi_i(R)\,R_i\,\varphi_i(\widehat R_i,R_{-i})$.  A mechanism is \textit{manipulable} if it is not strategyproof. It is \emph{nonbossy} if, for every profile $R$, every agent $i$, and every $R_i'\in\mathcal R_i$, $\varphi_i(R)=\varphi_i(R_i',R_{-i})$ implies $\varphi(R)=\varphi(R_i',R_{-i})$.

A pair $(i,c)$ \emph{blocks} a matching $\mu$ at profile $R$ if $c\,P_i\,\mu_i$ and either $|\mu(c)|<q_c$ or $i\succ_c j$ for some $j\in\mu(c)$. A matching is \emph{stable} at $R$ if it is individually rational and has no blocking pair. A mechanism is \emph{stable} if it selects a stable matching at every profile $R$.

At a profile $R$, a matching $\mu$ is a \emph{Pareto improvement over} a matching $\nu$ if $\mu_i\,R_i\,\nu_i$ for every agent $i$, with strict preference for some agent. A matching is \emph{Pareto efficient} if no feasible matching is a Pareto improvement over it, and a mechanism is Pareto efficient if it selects a Pareto efficient matching at every profile $R$. A mechanism $\psi$ \emph{Pareto-improves} a mechanism $\varphi$ if $\psi_i(R)\,R_i\,\varphi_i(R)$ for every agent $i$ and every profile $R$, with strict preference for some agent at some profile. In this case, $\psi$ is a \emph{Pareto improvement over} $\varphi$.

\section{Singleton-attainability and transparent access}\label{sec:sa}

\subsection{Option sets and direct access}\label{sec:option-sets}

Given the other agents' reports, an agent's option set is the set of objects that she can obtain by changing her own report (Barber\`a, 1983). 

\defhead{Option set and attainability} Fix an agent $i$ and the reports $R_{-i}$ of the other agents. Agent $i$'s \emph{option set} at $R_{-i}$ is
\[
O_i(R_{-i})=\{c\in C:\varphi_i(\widehat R_i,R_{-i})=c\text{ for some }\widehat R_i\in\mathcal R_i\}.
\]
An object $c$ is \emph{attainable} for $i$ at $R_{-i}$ if $c\in O_i(R_{-i})$. 
For an individually rational mechanism, the empty report always yields the outside option. Since our focus is access to objects, $O_i(R_{-i})$ records only attainable objects in $C$.
An object may be attainable only because other objects in the report alter the course of the mechanism. Our central concept rules out this dependence on the rest of the agent's report.

\defhead{Singleton-attainability (SA)} A mechanism $\varphi$ is \emph{singleton-attainable}  if, for every agent $i$, every report $R_{-i}$ of the other agents, and every object $c$,
\[
\varphi_i(\widehat R_i,R_{-i})=c\text{ for some }\widehat R_i\in\mathcal R_i
\quad\ \mbox{implies} \quad
\varphi_i((c),R_{-i})=c.
\]

Thus the option set can be recovered from singleton reports alone.

\begin{observation}[Singleton representation of option sets]\label{obs:oppset}
A mechanism $\varphi$ satisfies SA if and only if, for every agent $i$ and every $R_{-i}$,
\[
O_i(R_{-i})=\{c\in C:\varphi_i((c),R_{-i})=c\}.
\]
\end{observation}

\begin{proof}
Suppose $\varphi$ is SA. Then $c\in O_i(R_{-i})$ implies $\varphi_i((c),R_{-i})=c$. By the definition of the option set, $\varphi_i((c),R_{-i})=c$ implies $c\in O_i(R_{-i})$. Hence SA implies the displayed equality. Conversely, the equality immediately implies SA.
\end{proof}

Observation~\ref{obs:oppset} characterizes SA without any assumption concerning manipulability or strategyproofness. Strategyproofness imposes a different requirement: truthful reporting selects the agent's most-preferred acceptable object from $O_i(R_{-i})$ whenever it contains an acceptable object (Barber\`a, 1983). SA does not imply strategyproofness, so truthful reporting need not select the best attainable assignment using an SA mechanism. Its implication is that, once an attainable object $c$ has been identified, the agent need not compare $c$ with another acceptable object or report any other object as acceptable. 
Nonetheless, singleton-attainability should not be read as a recommendation that agents submit singleton reports. An agent may have several acceptable objects and good reasons to report them, especially when the reports of others are uncertain. 
Singleton-attainability is an access property: holding the other agents' reports fixed, every attainable object is accessible without declaring any other object acceptable.

SA alone does not imply that simplifying the report leaves others unaffected. Changing an agent's report may change the assignments of others. Nonbossiness strengthens the conclusion from individual access to replication of the entire outcome.

\begin{observation}[Singleton outcome replication]\label{obs:replication}
For an individually rational mechanism $\varphi$ which is SA and nonbossy, at every profile $R$, replacing any agent's report by the singleton report of her assigned object, or by the empty report if she is unmatched, leaves the entire matching unchanged.
\end{observation}

\begin{proof} Suppose $\varphi$ is individually rational, SA, and nonbossy. 
Fix a profile $R$ and an agent $i$. Suppose first that $i$ is matched at $R$, say $\varphi_i(R)=c$. Then $c$ is attainable for $i$ at $R_{-i}$, so SA gives $\varphi_i((c),R_{-i})=c$. 
Agent $i$'s assignment is the same at the original and singleton reports, hence nonbossiness implies that the entire matching is unchanged.
Now suppose that $i$ is unmatched at $R$. By individual rationality $i$ is unmatched at $(\emptyset,R_{-i})$ as well. The agent is unmatched regardless, and thus nonbossiness implies that the entire matching is unchanged in this case too. 
\end{proof}

Iterating this replacement, every matching generated by such a mechanism can be replicated by a report profile consisting only of singleton and empty reports.

\subsection{Profitable singleton manipulations} \label{sec:contraction-manip}

\defhead{Profitable manipulation} Given a mechanism $\varphi$, true preference $R_i$ and fixed $R_{-i}$, an alternative report $\widehat R_i$ is a \emph{profitable manipulation} for $i$ under $\varphi$ at $R_{-i}$ if $\varphi_i(\widehat R_i,R_{-i})\,P_i\,\varphi_i(R_i,R_{-i})$.

\medskip

Singleton-attainability is primarily a structural transparency
property, but it also has an immediate strategic implication. If a
profitable manipulation obtains an object $c$, SA implies that the
singleton report $(c)$ obtains the same object and is therefore also a
profitable manipulation. The converse holds as well.

\begin{proposition}[Singleton-attainability and profitable singleton
reports]\label{prop:singleton-manipulation}
A mechanism satisfies SA if and only if, whenever a profitable
manipulation $\widehat R_i$ for agent $i$ at true preference $R_i$
and fixed $R_{-i}$ obtains an object
$c=\varphi_i(\widehat R_i,R_{-i})$, the singleton report $(c)$ also obtains $c$ at $R_{-i}$.
\end{proposition}

\begin{proof}
Suppose first that $\varphi$ satisfies SA. Let $\widehat R_i$ be a
profitable manipulation at true preference $R_i$ and fixed $R_{-i}$,
and let
$c=\varphi_i(\widehat R_i,R_{-i})$. Since $\widehat R_i$ attains $c$,
SA implies
$\varphi_i((c),R_{-i})=c$. Thus, the singleton report $(c)$ obtains
the same object and is itself a profitable manipulation.

Conversely, suppose that the stated property holds. Fix an agent $i$,
$R_{-i}$, and an attainable object $c\in O_i(R_{-i})$, and choose a
report $\widehat R_i$ such that
$\varphi_i(\widehat R_i,R_{-i})=c$. Let the agent's true preference
be $R_i=(c)$. If $\varphi_i((c),R_{-i})\neq c$, then $c$ is strictly
preferred to the assignment obtained from the truthful report, so
$\widehat R_i$ is a profitable manipulation and obtains $c$. The
stated property would then imply
$\varphi_i((c),R_{-i})=c$. Thus the inequality cannot hold, and
$\varphi_i((c),R_{-i})=c$. Since $i$, $R_{-i}$, and the attainable
object $c$ were arbitrary, $\varphi$ satisfies SA.
\end{proof}

For individually rational mechanisms, the proposition implies
that any profitable manipulation under an SA mechanism can be
replaced by the singleton report of the target object, which is
truthfully acceptable. For a non-SA mechanism, by contrast, a
profitable manipulation may require a false acceptability claim
or a preference reversal.

Proposition~\ref{prop:singleton-manipulation} is related in spirit
to Kojima and Pathak's (2009) dropping-strategy result for the
multi-unit college side of stable matching. Their result reduces
arbitrary college manipulations to deletion-only reports that weakly
improve the college's bundle. Here, SA reduces any profitable
manipulation that obtains an object to a profitable singleton report.
Thus, within the individually rational SA class, a mechanism is
manipulable if and only if it admits a profitable singleton
manipulation.

\subsection{Best responses and Nash equilibrium}\label{sec:equilibrium}

Fix a true preference profile $R=(R_i)_{i\in I}$. The mechanism induces
a preference revelation game in which each agent chooses a report and
evaluates her assignment according to $R_i$. We call a unilateral
deviation from a submitted report profile \emph{profitable} if it gives the
deviating agent an assignment she strictly prefers to her assignment at
that profile. The proposition  below concerns only profitable deviations.

\begin{proposition}[Best responses and equilibrium verification]\label{prop:sa-best-response}
For an individually rational mechanism, the following statements are
equivalent:
\begin{enumerate}[label=(\roman*),nosep]
\item the mechanism satisfies SA;
\item for every agent $i$, every true preference $R_i$, and every
profile $R_{-i}$ of the other agents' reports, there exists a singleton
or empty best response;
\item a submitted report profile is a Nash equilibrium if and only if
no agent has a profitable singleton or empty deviation.
\end{enumerate}
\end{proposition}

\begin{proof}
This is the special case $b=1$ of Proposition~\ref{prop:best-response}.
\end{proof}

\section{Singleton-attainability, strategyproofness, and stability}\label{sec:decomp}

Singleton-attainability is related to both simple invariance requirements and some familiar
properties of mechanisms. The head-tail decomposition provides two invariance properties of mechanisms which are jointly sufficient for SA. Strategyproofness implies both
invariance properties, while stability yields SA separately.

\subsection{Head-tail decomposition}\label{sec:head-tail}

We use the following notation for two report modifications. Given a
report $\widehat R_i$ and object $c$,
$\, \widehat R_i\lift{c}$ moves $c$ to the top and preserves the relative 
order of every other listed object. For the second modification, 
$\widehat R_i\trunc{c}$ deletes every object ranked below $c$ if $c$ is acceptable,  and
leaves the rest unchanged; otherwise, 
$\widehat R_i\trunc{c}=\widehat R_i$. Thus,
$\widehat R_i\lift{c}=\widehat R_i$ when $c$ is already top-ranked,
while $\widehat R_i\trunc{c}=\widehat R_i$ when $c$ is unacceptable
or is the last-ranked acceptable object.

\defhead{Top-lift invariance (TLI)}
A mechanism $\varphi$ satisfies \emph{top-lift invariance} if, for
every $i$, $\widehat R_i$, $R_{-i}$,  and $c\in C$,
\[
\varphi_i(\widehat R_i,R_{-i})=c
\quad \mbox{implies} \quad
\varphi_i(\widehat R_i\lift{c},R_{-i})=c.
\]

Ehlers (2008) defines positive association by requiring that, if an agent is assigned an object and another object is ranked above it, interchanging the positions of these two objects in her report leaves her assignment unchanged. 
For individually rational mechanisms, taking the higher-ranked object
to be the one immediately above the assigned object and repeating
this operation shows that positive association implies TLI.

\defhead{Truncation invariance (TRI)}
A mechanism $\varphi$ satisfies \emph{truncation invariance} if, for
every $i$, $\widehat R_i$, $R_{-i}$, and $c\in C$,
\[
\varphi_i(\widehat R_i,R_{-i})=c
\quad \mbox{implies} \quad
\varphi_i(\widehat R_i\trunc{c},R_{-i})=c.
\]

When the obtained object is acceptable, call the objects ranked above
it the \emph{head}, and the objects ranked below it the \emph{tail}.
TLI preserves the assignment when the obtained object is lifted above
the entire head, while TRI preserves it when the entire tail is
deleted. 
The definitions also apply when the obtained object is unacceptable
in the original report: the top lift makes it top-ranked and
acceptable, and the subsequent truncation produces the singleton
report.

Truncation-based invariance conditions have been studied previously
in matching. Hatfield, Kominers, and Westkamp (2021) study
truncation consistency for deterministic mechanisms, and Zhang
(2023) studies truncation-invariance for random mechanisms.
On our deterministic strict-preference domain, both conditions
imply TRI, which considers only truncation immediately below
the assigned object.\footnote{Ehlers and Klaus (2014, 2016)
require the entire matching to remain unchanged after any
truncation that leaves the agent's original assignment acceptable.
See also Roth and Rothblum (1999) and Ehlers (2008) on truncation
strategies.}
Shirakawa (2025) studies both top-dropping and anti-bottom-dropping
monotonicity. For individually rational mechanisms, top-dropping
monotonicity implies TLI, while, for the truncation considered in
TRI, anti-bottom-dropping monotonicity rules out a change to another
object but permits a change to the outside option.

\begin{proposition}[The two invariances imply SA]
\label{prop:decomp}
A mechanism that satisfies TLI and TRI satisfies SA.
\end{proposition}

\begin{proof}
Suppose $\varphi$ satisfies TLI and TRI, and let
$\varphi_i(\widehat R_i,R_{-i})=c$. By TLI, lifting $c$ to the top preserves $c$. 
The object $c$
is therefore acceptable and top-ranked in $\widehat R_i\lift{c}$.
Subsequently truncating $\widehat R_i\lift{c}$ at $c$ deletes every other acceptable object and produces report $(c)$. By TRI, $c$ is preserved. Thus, $\varphi_i((c),R_{-i})=c$ and $\varphi$ satisfies SA. 
\end{proof}

Proposition~\ref{prop:decomp} is a sufficiency result, not a
characterization: the conjunction of TLI and TRI is sufficient but
not necessary for SA. Indeed, neither invariance is necessary for SA, as the example following Theorem~\ref{thm:smp} in Section~\ref{sec:smp} shows. 
The contrapositive of the proposition implies that every non-SA
mechanism violates TLI or TRI, or possibly both.
Table~\ref{tab:class} (Section~\ref{sec:classify}) records the
invariance properties of the primary mechanisms and families.

\subsection{Strategyproofness}\label{sec:strategyproofness}

One of the important sources of singleton-attainability is strategyproofness, which implies both invariances and hence SA.

\begin{proposition}[Strategyproofness implies both invariances and
singleton-attainability]\label{prop:sp-inv}
Every strategyproof mechanism satisfies TLI, TRI, and SA.
\end{proposition}

\begin{proof}
Suppose $\varphi$ is strategyproof. Fix $i$ and a profile $R$ with
$\varphi_i(R)=c$.

For TLI, let
$z=\varphi_i(R_i\lift{c},R_{-i})$. Then strategyproofness at $(R_i\lift{c},R_{-i})$ implies 
$z\,R_i\lift{c}\,c$.
Since $c$ is top-ranked in $R_i\lift{c}$, it follows that $z=c$.
Thus, $\varphi$ satisfies TLI.

For TRI, if $c\notin A_i(R_i)$, then
$R_i\trunc{c}=R_i$, so the conclusion is immediate. Suppose that
$c\in A_i(R_i)$, and let
$v=\varphi_i(R_i\trunc{c},R_{-i})$. Suppose that  $v\neq c$.
Then strategyproofness at $(R_i\trunc{c},R_{-i})$
requires $v\,P_i\trunc{c}\,c$, and hence $v\in A_i(R_i)$. 
Since $R_i\trunc{c}$ is a truncation of $R_i$ and preserves the relative ordering of all acceptable objects, it follows that $v \, P_i \, c$. This contradicts strategyproofness at the original preference profile $R$. Hence $v=c$ and  $\varphi$ satisfies TRI.

Since $\varphi$ satisfies TLI and TRI, Proposition~\ref{prop:decomp}  implies that $\varphi$ also satisfies SA.
\end{proof}

Proposition~\ref{prop:sp-inv} immediately implies that every non-SA
mechanism is manipulable. The converse does not hold: IA, for example, is SA but
manipulable. Combining Propositions~\ref{prop:singleton-manipulation}
and~\ref{prop:sp-inv}, an individually rational mechanism is
strategyproof if and only if it satisfies SA and admits no profitable
singleton manipulation. Moreover, for any SA mechanism, once
$R_{-i}$ is fixed, an agent only needs to test the singleton reports
of the objects ranked above her truthful assignment to identify every
object that she can obtain through a profitable manipulation.\footnote{See
Bonkoungou and Nesterov (2021) on objects reachable through profitable
strategies.}
Singleton-attainability also remains meaningful for strategyproof
mechanisms, where profitable manipulation is absent but the question
of how an agent accesses the objects in her option set remains.

\subsection{Stability}\label{sec:stable}

Stability provides another source of singleton-attainability, although stability alone does not imply TLI or TRI.

\begin{proposition}[Stability implies singleton-attainability]\label{prop:stable}
Every stable mechanism satisfies SA.
\end{proposition}

We will use the rural hospitals theorem repeatedly: at a fixed preference and priority profile, all stable matchings assign the same set of agents and the same number of agents is assigned to each object (Roth, 1986; Roth and Sotomayor, 1990).

\begin{proof} Suppose $\varphi$ is stable. Fix $i$ and a profile $R$ with
$\varphi_i(R)=c$. Let $\mu=\varphi(R)$. After $i$ replaces her report by the singleton report $(c)$, the matching $\mu$ remains stable: $i$ is assigned her only acceptable object in $\mu$ so $i$ is not in a blocking pair, and any blocking pair not involving $i$ would already have blocked $\mu$ at the original profile. Hence, there is a stable matching at $((c),R_{-i})$ in which $i$ is matched. By the rural hospitals theorem, $i$ is matched in every stable matching at this profile. Then, by individual rationality, $\varphi_i((c),R_{-i})=c$ and thus $\varphi$ is SA. 
\end{proof}

The conclusion of Proposition~\ref{prop:stable} extends beyond the school choice model. After a matched participant reports only her partner as acceptable, the original matching remains stable. 
The rural hospitals theorem fixes each participant's matched status across stable matchings at the modified profile, and individual rationality implies the same partner. This argument applies, for example, to participants on either side of the one-to-one two-sided matching model, known as the marriage problem (Gale and Shapley, 1962).

\section{Primary classification}\label{sec:classify}

Table~\ref{tab:class} presents the primary classification by singleton-attainability. Sections~\ref{sec:causes-sa} and~\ref{sec:failure} develop the structural explanations. Appendix~\ref{app:mechanisms} collects definitions and references for the primary mechanisms and families central to the analysis. 
The Top Trading Cycles (TTC) entry in Table~\ref{tab:class} (and in Table~\ref{tab:classes} in Section~\ref{sec:class-revisited})  includes TTC and its generalizations.\footnote{TTC includes the Shapley--Scarf (1974) TTC mechanism, the Trading Cycles mechanisms of Pycia and \"{U}nver (2017), including the Hierarchical Exchange mechanisms of P\'apai (2000), the school choice TTC mechanisms of Abdulkadiro\u{g}lu and S\"onmez (2003), and school choice variants such as the Clinch and Trade mechanisms of Morrill (2015) and the Equitable TTC mechanism of Hakimov and Kesten (2018). These mechanisms are strategyproof, as shown by Roth (1982b) for the original TTC mechanism and by the respective cited
papers for the remaining mechanisms.} 
Some entries in the table overlap. In particular, DA and IA belong to both the PRP and Taiwan Deduction families.

Table~\ref{tab:class} reveals three interesting contrasts. First, DA, Serial Dictatorship, and TTC satisfy strategyproofness, while EADA and DA-TTC are Pareto improvements over the strategyproof DA mechanism. 
Second, PRP and Taiwan Deduction mechanisms use report-induced priorities that reflect ordinary priority improvement, while Rank-Equity reverses this direction by construction. 
Third, within each rank-welfare pair, the fixed and reported versions
differ in how the outside option is ranked: its rank is fixed in the
former and determined by the submitted report in the latter. The two
versions use the same rank-welfare criterion on complete reports but fall
on opposite sides of SA when incomplete reports are allowed.

\begin{table}[!t]
\centering
\scriptsize
\setlength{\tabcolsep}{4pt}
\renewcommand{\arraystretch}{1.08}

\newcommand{\SAclassificationHalf}[1]{%
  \begin{minipage}[t]{\linewidth}\vspace{0pt}%
  \begin{tabularx}{\linewidth}{@{}>{\raggedright\arraybackslash}X
    >{\raggedright\arraybackslash}p{0.43\linewidth}@{}}
  #1
  \end{tabularx}%
  \end{minipage}%
}

\begin{tabular}{@{}p{0.475\textwidth}@{\hspace{0.55em}}|
  @{\hspace{0.55em}}p{0.475\textwidth}@{}}
\toprule
{\centering\textbf{Singleton-attainable}\par}
& {\centering\textbf{Not singleton-attainable}\par} \\
\midrule
\SAclassificationHalf{%
  \textit{Primary mechanism or family} & \textit{Reason} \\
}
&
\SAclassificationHalf{%
  \textit{Primary mechanism or family} & \textit{Reason} \\
} \\
\midrule

\SAclassificationHalf{%
  DA & TLI and TRI\newline
    (Proposition~\ref{prop:sp-inv}) \\
  \midrule
  Serial Dictatorship & TLI and TRI\newline
    (Proposition~\ref{prop:sp-inv}) \\
  \midrule
  TTC & TLI and TRI\newline
    (Proposition~\ref{prop:sp-inv}) \\
}
&
\SAclassificationHalf{%
  EADA & TLI, not TRI:\newline
    tail-dependent\newline
    (Example~\ref{ex:da-improvements-ti};
    Proposition~\ref{prop:complementary-invariance}) \\
  \midrule
  DA-TTC & TLI, not TRI:\newline
    tail-dependent\newline
    (Example~\ref{ex:da-improvements-ti};
    Proposition~\ref{prop:complementary-invariance}) \\
} \\
\midrule

\SAclassificationHalf{%
  Object-Proposing DA & TLI and TRI\newline
    (Corollary~\ref{cor:smp-applications}) \\
  \midrule
  IA & TLI and TRI\newline
    (Corollary~\ref{cor:smp-applications}) \\
  \midrule
  PRP mechanisms & TLI and TRI\newline
    (Corollary~\ref{cor:smp-applications}) \\
  \midrule
  Taiwan Deduction\newline mechanisms & TLI and TRI\newline
    (Corollary~\ref{cor:smp-applications}) \\
}
&
\SAclassificationHalf{%
  Rank-Equity & Not TLI, but TRI:\newline
    head-dependent\newline
    (Example~\ref{ex:rank-equity};
    Proposition~\ref{prop:complementary-invariance}) \\
} \\
\midrule

\SAclassificationHalf{%
  Fixed Rawlsian\newline mechanisms & TLI and TRI\newline
    (Corollary~\ref{cor:fixed-rank-welfare-applications}) \\
}
&
\SAclassificationHalf{%
  Reported Rawlsian\newline mechanisms & TLI, not TRI:\newline
    tail-dependent\newline
    (Example~\ref{ex:reported};
    Proposition~\ref{prop:complementary-invariance}) \\
} \\
\midrule
\SAclassificationHalf{%
  Fixed Additive Rank\newline Scoring mechanisms & TLI and TRI\newline
    (Corollary~\ref{cor:fixed-rank-welfare-applications}) \\
}
&
\SAclassificationHalf{%
  Reported Additive Rank\newline Scoring mechanisms & TLI, not TRI:\newline
    tail-dependent\newline
    (Example~\ref{ex:reported};
    Proposition~\ref{prop:complementary-invariance}) \\
} \\
\midrule
\SAclassificationHalf{%
  Fixed Rank-Maximal\newline mechanisms & TLI and TRI\newline
    (Corollary~\ref{cor:fixed-rank-welfare-applications}) \\
}
&
\SAclassificationHalf{%
  Reported Rank-Maximal\newline mechanisms & TLI, not TRI:\newline
    tail-dependent\newline
    (Example~\ref{ex:reported};
    Proposition~\ref{prop:complementary-invariance}) \\
} \\
\bottomrule
\end{tabular}
\caption{Primary classification by singleton-attainability}
\label{tab:class}
\end{table}

The \textit{Reason} columns record the invariance properties of the mechanisms and
indicate which example or result establishes them. 
TLI eliminates dependence
on the head by allowing the assigned object to be lifted above it, while
TRI eliminates dependence on the tail by allowing it to be deleted (Section~\ref{sec:head-tail}). An SA mechanism
need not satisfy either invariance, as the example in
Section~\ref{sec:smp} shows.
Remarkably, all
singleton-attainable primary mechanisms and families satisfy both TLI and
TRI.
For the entries on the left side of the table, the displayed TLI and
TRI properties imply SA by Proposition~\ref{prop:decomp}.

The contrapositive of Proposition~\ref{prop:decomp} indicates that a
non-SA mechanism must violate at least one of the two invariances.
We call a non-SA mechanism \emph{head-dependent} if it satisfies TRI
but violates TLI, and \emph{tail-dependent} if it satisfies TLI
but violates TRI. The restriction to non-SA mechanisms ensures that
these terms describe forms of indirect access, since an SA mechanism
need not satisfy either invariance. Note also that a non-SA mechanism may
violate both TLI and TRI, so it need not be either head-dependent
or tail-dependent, both of which indicate one-sided invariance.
In our classification, each non-SA mechanism violates exactly one
of TLI and TRI, so each is either head-dependent or tail-dependent.

The primary classification also illustrates that SA is
weaker than strategyproofness. Object-Proposing DA, IA, and the three primary
fixed rank-welfare families are SA but manipulable. Within each of
the PRP and Taiwan Deduction families, DA is the only strategyproof
mechanism.\footnote{See
Roth (1982a) for Object-Proposing DA, Ayoade and P\'apai (2023) for
PRP mechanisms, Dur et al. (2022) for Taiwan Deduction mechanisms,
and Featherstone (2020) for Fixed Additive Rank Scoring and Fixed
Rank-Maximal mechanisms. For Fixed Rawlsian mechanisms, the
manipulation argument in Afacan and Dur (2024, Proposition~5)
remains valid when the outside option rank is fixed at $m+1$.}
Every mechanism listed on the non-SA side of
Table~\ref{tab:class} is manipulable by
Proposition~\ref{prop:sp-inv}.

\section{What causes singleton-attainability?}\label{sec:causes-sa}

Strategyproofness is one broad source of singleton-attainability and we treat the class of strategyproof mechanisms as our first structural class. 
Here we present two further sources of SA. The first one generalizes stability to a larger class of mechanisms that select stable matchings with respect to specific report-induced priorities. The second one is based on rank-welfare optimization when the social treatment of the outside option is exogenous.

\subsection{Stable Monotone Priority mechanisms}\label{sec:smp}

One source of SA is stability with respect to report-induced priorities that satisfy the following natural monotonicity condition: reducing the set of objects ranked above an acceptable object cannot weaken the agent's priority at that object. The resulting Stable Monotone Priority class guarantees SA for arbitrary stable selection, while its agent-optimal and object-optimal selections satisfy the stronger TLI and TRI properties. This structural class extends substantially beyond fixed-priority stable mechanisms.\footnote{Sasaki (2025) studies DA with rank-dependent priorities. The motivating Hyogo high school admission mechanism is a two-tier Taiwan Deduction mechanism, while his general rank-monotone framework is broader.}

For every agent $i$, report $R_i$, and object $c\in A_i(R_i)$, define the \textit{upper contour set} of $c$ by
\[
U_i(c;R_i)=\{d\in A_i(R_i):d\,P_i\,c\}.
\]

A \emph{report-induced priority rule} $\Gamma=(\Gamma_c)_{c\in C}$ assigns, to every report profile $R$ and object $c$, a strict priority order $\Gamma_c(R)$ over the agents. The rule $\Gamma$ may also depend on fixed primitives or mechanism parameters, such as exogenous priorities, scores, and tie-breakers, which are suppressed from the notation.

\defhead{Independence} A report-induced priority rule $\Gamma$ satisfies \emph{independence} if, for every agent $i$, object $c$, reports $R_i$ and $R_i'$, profile $R_{-i}$, and pair of distinct agents $j,h\in I\setminus\{i\}$,
\[
j\,\Gamma_c(R_i,R_{-i})\,h
\quad  \mbox{ if and only if  }  \quad
j\,\Gamma_c(R_i',R_{-i})\,h.
\]
Thus, changing agent $i$'s report does not change the relative priority of any two other agents at any object.

\medskip

The notion of priority monotonicity, which we define next, is related to the familiar property of respecting improvements (Balinski and S\"onmez, 1999), but applies to the report-induced priority at a particular object rather than to the mechanism's final assignment.

\defhead{Priority monotonicity} A report-induced priority rule $\Gamma$ satisfies \emph{priority monotonicity} if, for every agent $i$, object $c$, reports $R_i$ and $R_i'$, and profile $R_{-i}$, if $c\in A_i(R_i)$, $c\in A_i(R_i')$, and
\[
U_i(c;R_i')\subseteq U_i(c;R_i),
\]
then, for every $j\neq i$,
\[
i\,\Gamma_c(R_i,R_{-i})\,j
\quad \mbox{
implies }  \quad
i\,\Gamma_c(R_i',R_{-i})\,j.
\]

\medskip

Thus, reducing the set of objects ranked above $c$ preserves $i$'s priority over every agent whom she previously outranked at $c$. 
Priority monotonicity allows an agent's priority at $c$ to depend on which objects she ranks above $c$, not merely on their number. For example, a score rule may deduct a fixed nonnegative amount for each object ranked above $c$, with different deductions for different objects. Removing objects from the head then weakly raises the score, 
but two reports placing $c$ at the same rank with different heads may yield different scores and hence different priorities.

\defhead{Structural Class: Stable Monotone Priority}
A mechanism $\varphi$ is a \textit{Stable Monotone Priority mechanism} if there is a report-induced priority rule $\Gamma$ satisfying independence and priority monotonicity such that, for every profile $R$, $\varphi(R)$ is stable at $R$ with respect to the priority profile $\Gamma(R)$.

\begin{theorem}[Stable Monotone Priority mechanisms]\label{thm:smp}
Every Stable Monotone Priority mechanism satisfies SA.
\end{theorem}

\begin{proof} Let $\varphi$ be a Stable Monotone Priority mechanism with report-induced priority rule $\Gamma$.
Fix $i$, $R_{-i}$, an object $c\in C$, and a report $\widehat R_i$
such that $\varphi_i(\widehat R_i,R_{-i})=c$.
Write $\widehat R=(\widehat R_i,R_{-i})$, $R'=((c),R_{-i})$, and $\mu=\varphi(\widehat R)$. By definition, $\mu$ is stable at $\widehat R$ with respect to $\Gamma(\widehat R)$ and assigns $c$ to $i$.

We first show that $\mu$ remains stable at $R'$ with respect to $\Gamma(R')$. Individual rationality is preserved because $i$ remains assigned to $c$, her only acceptable object. Moreover, a new blocking pair cannot involve $i$ since she obtains her first choice. Suppose there is a blocking pair $(j,d)$ with $j\neq i$. Since $j$'s preference is unchanged, if $|\mu(d)|<q_d$, then $(j,d)$ would already have been a blocking pair for $\mu$ at $\widehat R$. Thus $|\mu(d)|=q_d$.

Independence implies that every priority comparison between $j$ and an agent $h\in\mu(d)\setminus \{i\}$ is unchanged between $\Gamma_d(\widehat R)$ and $\Gamma_d(R')$. If $d\neq c$, then $i\notin\mu(d)$, so every relevant comparison is unchanged and $(j,d)$ would already have blocked $\mu$ at $\widehat R$. If $d=c$, given
$U_i(c;(c))=\emptyset\subseteq U_i(c;\widehat R_i)$,
priority monotonicity implies that if
$i\,\Gamma_c(\widehat R)\,j$ then 
$i\,\Gamma_c(R')\,j$. 
By independence, $(j,c)$ could be a new blocking pair only if
$j\,\Gamma_c(R')\,i$. Then, by priority monotonicity,  $j\,\Gamma_c(\widehat R)\,i$.
Hence $(j,c)$ would already have blocked $\mu$ at $\widehat R$
with respect to $\Gamma(\widehat R)$.
In every case we obtain a contradiction to the stability of $\mu$ at $\widehat R$ with respect to $\Gamma(\widehat R)$. Thus, $\mu$ remains stable at $R'$ with respect to $\Gamma(R')$.

This means that there is a matching that is stable at $R'$ with respect to $\Gamma(R')$ and matches agent $i$. Then, by the rural hospitals theorem, every stable matching at $R'$ with respect to $\Gamma(R')$ matches $i$, including $\varphi(R')$. 
Given the singleton report, individual rationality implies that
$\varphi_i(R')=\varphi_i((c),R_{-i})=c$.
Thus, Stable Monotone Priority mechanisms satisfy SA.
\end{proof}

Fixed-priority stable mechanisms are immediate members of this structural class: for every profile $R$ and each object $c $, let $\Gamma_c(R)=\, \succ_c$. Both independence and priority monotonicity hold trivially, so Proposition~\ref{prop:stable} is recovered as a special case of Theorem~\ref{thm:smp}.

For Stable Monotone Priority mechanisms, no restriction is imposed on how $\varphi$ selects among stable matchings. Thus, by Theorem~\ref{thm:smp}, an arbitrary stable selection is sufficient for SA. However, it does not guarantee TLI or TRI. To see this, consider three agents and three unit-capacity objects,
with reports $R_1=(a,b,c)$, $R_2=(b,c,a)$, and $R_3=(c,a,b)$,
and fixed priorities $\succ_a \, = (2,3,1)$, $\succ_b \, = (3, 1, 2)$,
and $\succ_c \,  = (1, 2, 3)$. Let a
mechanism select $(1\asgn b, 2 \asgn c, 3 \asgn a)$ at this profile, $(1\asgn c, 2 \asgn a, 3 \asgn b)$ after agent~1
lifts $b$ to the top, and $(1\asgn a, 2 \asgn b, 3 \asgn c)$ after she truncates her report below $b$ at the original profile.
At all other profiles, let it select DA. This mechanism is stable
and hence satisfies SA by Proposition~\ref{prop:stable}, but violates
both TLI and TRI.

We show next that TLI and TRI follow for either of the two extremal
stable selections of Stable Monotone Priority mechanisms. This result
then allows us to deduce both invariances for the primary mechanisms
and families in the Stable Monotone Priority class.

\begin{proposition}[Extremal stable selection]\label{prop:extremal-selection}
Let $\varphi$ be a Stable Monotone Priority mechanism with report-induced priority rule $\Gamma$. If $\varphi$ selects the agent-optimal stable matching at every profile $R$ with respect to $\Gamma(R)$, or the object-optimal stable matching at every profile $R$ with respect to $\Gamma(R)$, then $\varphi$ satisfies TLI and TRI.
\end{proposition}

\begin{proof}
See Appendix~\ref{app:proof-extremal-selection}.
\end{proof}

\begin{corollary}[Primary mechanisms and families in the Stable Monotone Priority class]\label{cor:smp-applications}
Object-Proposing DA, IA, PRP mechanisms, and Taiwan Deduction mechanisms satisfy TLI, TRI, and SA.
\end{corollary}

\begin{proof}

We show that each mechanism is a Stable Monotone Priority mechanism
and uses an extremal stable selection. Then
Proposition~\ref{prop:extremal-selection} provides the invariance
properties, and Proposition~\ref{prop:decomp} implies SA.
Since all these mechanisms belong to the Stable
Monotone Priority class, SA also follows directly by
Theorem~\ref{thm:smp}.

\noindent \textit{Object-Proposing DA.} Object-Proposing DA uses fixed priorities, so independence and priority monotonicity hold trivially. It is a Stable Monotone Priority mechanism and selects the object-optimal stable matching. 

\noindent \textit{PRP mechanisms.} For PRP mechanisms, changing
$i$'s report changes only $i$'s position in each induced priority
order, so independence holds. Reducing the head of $c$ weakly
improves the preference-rank class in which $i$ places $c$.
Thus priority monotonicity holds.

\noindent \textit{Taiwan Deduction mechanisms.} For Taiwan
Deduction mechanisms, changing $i$'s report leaves every other
agent's adjusted score unchanged, so independence holds.
Reducing the head of $c$ weakly reduces the deduction applied
to $c$ and hence weakly increases $i$'s adjusted score at $c$.
Thus priority monotonicity holds.

Both the PRP and Taiwan Deduction mechanisms apply DA to their induced priorities and select the agent-optimal stable matching. Finally, note that IA is a member of both families.
\end{proof}

\subsection{Fixed Rank Welfare mechanisms}\label{sec:fixed-outside}

Another source of SA is rank-welfare optimization when the rank of the outside option is
fixed independently of the report. The resulting Fixed Rank Welfare class, which encompasses various assignment-rank objectives, satisfies TLI and
TRI, and hence SA, by Proposition~\ref{prop:decomp}.

For agent $i$, preference $R_i$, and object $c\in A_i(R_i)$, let $\rk_i(c;R_i)$ denote the position of $c$ in $R_i$. In an individually rational matching $\mu$, agent $i$'s \emph{fixed assignment
rank} $\rho_i^F(\mu_i;R_i)$ is $\rk_i(\mu_i;R_i)$ if she is matched, and $m+1$ if she
receives the outside option.  This is the \textit{fixed} convention for treating the outside option, so shortening a report does not change the rank
assigned to it. 
Write $\rho^F(\mu;R)=(\rho_i^F(\mu_i;R_i))_{i\in I}$ for
the fixed assignment rank list at report profile $R$.

A \emph{rank list} is a vector $x\in\{1,\ldots,m+1\}^I$.
For any such list, let $x^\rightarrow$ denote its coordinates arranged
in weakly increasing order, so $x_1^\rightarrow \le x_2^\rightarrow \le \cdots \le x^\rightarrow_{|I|} $. Keeping in mind that lower ranks are preferred, we will use
the following notion of dominance for two rank lists that are unambiguously comparable.

\defhead{Rank dominance}
For two rank lists $x$ and $y$, say that $x$ \emph{rank-dominates} $y$, written $x\succeq_{\mathrm{rk}} y$, if
\[
x^\rightarrow_\ell\le y^\rightarrow_\ell
\qquad\text{for every }\ell=1,\ldots,|I|.
\]

Rank dominance does not compare every pair of rank lists.
We define next complete welfare orderings of rank lists
that respect rank dominance.

\defhead{Rank monotonicity}
A complete and transitive welfare ordering $\unrhd$ of rank lists is
\emph{rank-monotonic} if $x\succeq_{\mathrm{rk}}y$ implies $x\unrhd y$.

\defhead{Strictness requirement}
A welfare ordering $\unrhd$ satisfies the \emph{strictness requirement}
if, for all rank lists $x$ and $y$, if $x\succeq_{\mathrm{rk}}y$ and
$x$ has fewer entries equal to $m+1$ than $y$ then  $\unrhd$
ranks $x$ strictly above $y$. 

\medskip

This means that whenever \(x\succeq_{\mathrm{rk}}y\), the welfare ordering weakly prefers \(x\) to \(y\), with strict preference if \(x\) has fewer entries equal to \(m+1\) than \(y\).
Under the fixed assignment rank
convention, entries equal to $m+1$ correspond exactly to unmatched
agents. Thus, a rank-dominating list that leaves fewer agents
unmatched is strictly preferred by $\unrhd$.

We will say that an individually rational matching $\mu$ is
\textit{welfare-maximal} at profile $R$ if, for every individually
rational matching $\nu$ at $R$,
$\rho^F(\mu;R)\unrhd\rho^F(\nu;R)$. Fix a strict report-independent
ordering $\sigma$ of feasible matchings to break welfare ties.

\defhead{Structural Class: Fixed Rank Welfare}
A mechanism $\varphi$ is a \textit{Fixed Rank Welfare mechanism} if
there exists a rank-monotonic welfare ordering $\unrhd$ satisfying
the strictness requirement such that, at every profile $R$,
$\varphi(R)$ is the $\sigma$-highest among the individually rational
matchings $\mu$ whose fixed assignment-rank list
$\rho^F(\mu;R)$ is maximal according to $\unrhd$.

\medskip

Thus, at every profile, the selected matching is welfare-maximal when
welfare is evaluated using fixed assignment ranks, and it is the
$\sigma$-highest among all welfare-maximal matchings.

\begin{theorem}[Fixed Rank Welfare]\label{thm:fixed-outside}
Every Fixed Rank Welfare mechanism satisfies TLI, TRI, and SA.
\end{theorem}

\begin{proof}
Let $\varphi$ be a Fixed Rank Welfare mechanism with welfare ordering $\unrhd$. 
Fix agent $i$, $R_{-i}$, an object $c\in C$, and a report
$\widehat R_i$ such that $\varphi_i(\widehat R_i,R_{-i})=c$.
Write $\widehat R=(\widehat R_i,R_{-i})$ and $\mu=\varphi(\widehat R)$. 
By definition, $\mu$ is the $\sigma$-highest individually rational welfare-maximal matching at $\widehat R$ according to \(\unrhd\) and assigns $c$ to $i$. In particular, $c\in A_i(\widehat R_i)$.

\emph{TLI.} Let $R'=(\widehat R_i^{\uparrow c},R_{-i})$ and $\nu=\varphi(R')$. Matching $\mu$ remains individually rational at $R'$. Lifting $c$ weakly improves the fixed assignment rank of $i$ in $\mu$ and leaves every other agent's fixed assignment rank unchanged. Hence $\rho^F(\mu;R')\succeq_{\mathrm{rk}}\rho^F(\mu;\widehat R)$, and thus rank monotonicity implies $\rho^F(\mu;R')\unrhd\rho^F(\mu;\widehat R)$.

Suppose, toward a contradiction, that $\nu_i\neq c$. Since the lift does not change acceptability, $\nu$ is also individually rational at $\widehat R$.
If $\nu_i$ is an object ranked above $c$ in $\widehat R_i$,
lifting $c$ increases its assignment rank by one.
If $\nu_i$ is ranked below $c$  in $\widehat R_i$, its assignment rank is unchanged, while if $\nu_i=\emptyset$, its fixed assignment rank remains $m+1$. All other agents' ranks are unchanged. Hence $\rho^F(\nu;\widehat R)\succeq_{\mathrm{rk}}\rho^F(\nu;R')$, and thus rank monotonicity implies $\rho^F(\nu;\widehat R)\unrhd\rho^F(\nu;R')$. Since $\mu$ is welfare-maximal at $\widehat R$ according to $\unrhd$, we have $\rho^F(\mu;\widehat R)
\unrhd\rho^F(\nu;\widehat R)$. Putting these together, 
\[
\rho^F(\mu;R')\unrhd\rho^F(\mu;\widehat R)
\unrhd\rho^F(\nu;\widehat R)\unrhd\rho^F(\nu;R').
\]
But $\nu$ is welfare-maximal  at $R'$ and $\mu$ is individually rational at the same profile, so $\rho^F(\nu;R')\unrhd\rho^F(\mu;R')$. Hence all the preceding comparisons are indifferences. In particular, $\nu$ is welfare-maximal at $\widehat R$, and $\mu$ is welfare-maximal at $R'$. Since $\mu\neq\nu$ and $\mu$ is selected at $\widehat R$, the ordering $\sigma$ ranks $\mu$ above $\nu$. Moreover, given that $\sigma$ is report-independent, $\mu$ must be selected over $\nu$ at $R'$, a contradiction. Therefore, $\varphi_i(\widehat R_i^{\uparrow c},R_{-i})=c$, and $\varphi$ satisfies TLI.

\emph{TRI.} Let $R''=(\widehat R_i^{\le c},R_{-i})$. Matching $\mu$ remains individually rational at $R''$. Every individually rational matching at $R''$ assigns $i$ either $c$, an object ranked above $c$ in $\widehat R_i$, or the outside option. Truncating below $c$ leaves the assignment rank of $c$ and of every object ranked above $c$ unchanged, while the outside option retains its fixed rank $m+1$. Hence every matching that remains individually rational has the same fixed assignment-rank list at $R''$ as at $\widehat R$, while matchings assigning $i$ a deleted object are no longer individually rational. Thus, $\mu$ remains welfare-maximal at $R''$. Moreover, any matching that ties $\mu$ at $R''$ also tied $\mu$ at $\widehat R$. Since $\sigma$ is report-independent and selected $\mu$ over every such matching at $\widehat R$, it also selects $\mu$ at $R''$. Hence, $\varphi_i(\widehat R_i^{\le c},R_{-i})=c$. Therefore, $\varphi$ satisfies TRI.

Proposition~\ref{prop:decomp} therefore implies that $\varphi$
satisfies SA.
\end{proof}

Three of the primary families of mechanisms belong to the class of Fixed Rank Welfare mechanisms. 
Fixed Rawlsian mechanisms first minimize the worst fixed assignment rank and then the
number of agents receiving that rank. Fixed Additive Rank Scoring mechanisms minimize
the sum of a strictly increasing function of the fixed assignment ranks. 
Fixed Rank-Maximal mechanisms lexicographically maximize the numbers
of agents receiving fixed assignment ranks $1,2,\ldots,m+1$, in that order. Each family uses a report-independent tie-break.

\begin{corollary}[Primary mechanisms and families in the Fixed Rank Welfare
class]\label{cor:fixed-rank-welfare-applications}
Fixed Rawlsian mechanisms, Fixed Additive Rank Scoring mechanisms, and Fixed
Rank-Maximal mechanisms satisfy TLI, TRI, and SA.
\end{corollary}

\begin{proof}
Fix rank lists $x$ and $y$ with $x\succeq_{\mathrm{rk}}y$, so
$x^\rightarrow_\ell\le y^\rightarrow_\ell$ for every
$\ell=1,\ldots,|I|$.
For each of the three primary families, we show that its welfare
ordering weakly prefers $x$ to $y$, which gives rank monotonicity,
and strictly prefers $x$ to $y$ whenever $x$ has fewer entries equal
to $m+1$ than $y$, which verifies the strictness requirement.

\smallskip

\noindent \textit{Fixed Rawlsian mechanisms.}
The inequalities $x^\rightarrow_\ell\le y^\rightarrow_\ell$ imply
$\max_{i\in I}x_i\le\max_{i\in I}y_i$. If the latter inequality is
strict, the Rawlsian ordering strictly prefers $x$ to $y$.
If $\max_{i\in I}x_i=\max_{i\in I}y_i=r$, every position $\ell$
with $x^\rightarrow_\ell=r$ also satisfies $y^\rightarrow_\ell=r$, 
since $r = x^\rightarrow_\ell\le y^\rightarrow_\ell \leq \max_{i\in I}y_i= r$,
so $x$ has no more entries equal to $r$ than $y$. Hence $x\unrhd y$.
Now suppose that $x$ has fewer entries equal to $m+1$ than $y$.
If $x$ has no such entries, its largest entry is at most $m$,
whereas the largest entry of $y$ is $m+1$. Otherwise, both lists
have largest entry $m+1$, but $x$ has fewer entries at that rank.
In either case, the Rawlsian ordering strictly prefers $x$ to $y$.

\smallskip

\noindent \textit{Fixed Additive Rank Scoring mechanisms.}
Since the scoring function $g$ is strictly increasing,
$g(x^\rightarrow_\ell)\le g(y^\rightarrow_\ell)$ for every $\ell$.
Summing these inequalities and noting that sorting does not change
the total score gives
$\sum_{i\in I}g(x_i)\le\sum_{i\in I}g(y_i)$.
Hence $x\unrhd y$. If $x$ has fewer entries equal to $m+1$ than $y$,
at least one of the inequalities $x^\rightarrow_\ell\le y^\rightarrow_\ell$
is strict, so the total
score of $x$ is strictly lower.

\smallskip

\noindent \textit{Fixed Rank-Maximal mechanisms.}
For every rank $r$, the inequalities $x^\rightarrow_\ell\le y^\rightarrow_\ell$ for every $\ell=1,\ldots,|I|$ imply
$|\{i\in I:x_i\le r\}|\ge|\{i\in I:y_i\le r\}|$.
If the rank-count lists differ, let $r$ be the first rank at which
their counts differ. The counts at all smaller ranks coincide, so
$x$ must have more entries equal to $r$ than $y$.
Thus $x\unrhd y$, with strict preference whenever the rank-count
lists differ. 
If \(x\) has fewer entries equal to \(m+1\) than \(y\), their rank counts must differ at some rank at most \(m\). The preceding argument shows that the first such difference favors \(x\), so the welfare comparison is strict.

\smallskip

Thus, all three welfare orderings are rank-monotonic, and the three
primary families are Fixed Rank Welfare mechanisms.
By Theorem~\ref{thm:fixed-outside}, they satisfy TLI, TRI, and SA.
\end{proof}

\section{Report width and boundedness}\label{sec:width}

Singleton-attainability extends naturally to a measure of reporting burden, which we call \emph{report width}. To determine report width across market sizes, we make the market explicit. From now on, $\varphi$ denotes a mechanism applied across all finite markets. A market $E$ specifies its agent and object sets, capacities and, when relevant,  priorities. Let $\mathcal E_m$ denote the class of such markets with exactly $m$ objects, and write $I(E)$ for the agent set of $E$. When a market $E$ is fixed, we suppress it from the notation. A mechanism property is understood to hold for $\varphi$ if it holds in every such market.

Fix a mechanism $\varphi$,  $\, m\ge1$, a market $E\in\mathcal E_m$, an agent $i\in I(E)$, a profile $R_{-i}$ in $E$, and an attainable object $c\in O_i(R_{-i})$. 
The number of objects declared acceptable in report $\widehat R_i$ is $|A_i(\widehat R_i)|$. The \emph{minimum report size} needed to attain $c$ at $R_{-i}$  under mechanism $\varphi$ is
$$
k_i^\varphi(c;R_{-i})
=
\min\bigl\{|A_i(\widehat R_i)|:\varphi_i(\widehat R_i,R_{-i})=c\bigr\}.
$$
Thus $k_i^\varphi(c;R_{-i})$ is the length of a shortest report with which $i$ can obtain $c$ under $\varphi$, given $R_{-i}$.

\defhead{Report width}
For a mechanism $\varphi$ and markets with exactly $m$ objects, the
\emph{report width} $k_m(\varphi)$ is the largest minimum report size
that can arise, normalized to be at least one.\footnote{To simplify the exposition, the
normalization ensures that report-width bounds are positive. It affects only the degenerate cases in which no object is
attainable in any market in $\mathcal E_m$, or every attainable object
can already be obtained from the empty report. The normalization convention only imposes that these cases be recorded as
width one rather than zero.}
More precisely, it is the maximum of $1$ and all values of
$k_i^\varphi(c;R_{-i})$ across all markets $E\in\mathcal E_m$, agents
$i\in I(E)$, profiles $R_{-i}$ in $E$, and attainable objects
$c\in O_i(R_{-i})$.

The \emph{report width} of mechanism $\varphi$ (or \textit{width}, for short) is
$$
k(\varphi)=\sup_{m\ge1}k_m(\varphi).
$$

Report width measures the support size of a shortest access report, subject to the normalization convention. 
Therefore, $k_m(\varphi)=1$ for all $m \ge 1$ if and only if $k(\varphi)=1$.
Mechanism $\varphi$ has \emph{bounded width} if $k(\varphi)<\infty$ and \emph{unbounded width} otherwise. It has \emph{maximal width} if $k_m(\varphi)=m$ for every $m\ge1$. Width statements for a mechanism family or class apply to every
member.

\medskip

\begin{observation}[Width one and singleton-attainability]
\label{obs:width-one}
Every SA mechanism has report width one. Conversely, every
individually rational mechanism with report width one satisfies SA.
\end{observation}

\begin{proof}
Suppose that $\varphi$ satisfies SA. Every attainable object can be
obtained from its singleton report, so for every $m \ge 1, \, k_m(\varphi) \le 1$. Then, by the normalization convention, for every $m\ge1, \, 
k_m(\varphi)=1$, and hence $k(\varphi)=1$.

Conversely, suppose that $\varphi$ is individually rational and
$k(\varphi)=1$. Fix an agent $i$, a profile $R_{-i}$, and an
attainable object $c\in O_i(R_{-i})$. Then
$k_i^\varphi(c;R_{-i})\leq1$, so some report declaring at most one
object acceptable assigns $c$ to $i$. This report cannot be empty,
since individual rationality implies that the empty report assigns
the outside option. It therefore declares exactly one object
acceptable. Since it assigns $c$, individual rationality implies that the report is $(c)$. This proves that $\varphi$ is SA.
\end{proof}

Thus, SA and report width one are equivalent for individually rational
mechanisms. Moreover, every strategyproof mechanism satisfies SA by
Proposition~\ref{prop:sp-inv} and therefore has report width one by
Observation~\ref{obs:width-one}.

For $b\ge1$, let $O_i^b(R_{-i})$ be agent $i$'s \textit{constrained option set} when $i$'s
report is restricted to at most $b$ acceptable objects, given $R_{-i}$. These constrained option sets are nested:
$$
O_i^1(R_{-i})\subseteq O_i^2(R_{-i})\subseteq\cdots\subseteq O_i^m(R_{-i})=O_i(R_{-i}).
$$

By construction, $k_m(\varphi)$ is the smallest positive integer
$b$ such that
$O_i^b(R_{-i})=O_i(R_{-i})$ for every market
$E\in\mathcal E_m$, every agent $i\in I(E)$, and every profile
$R_{-i}$ in $E$.

If $\varphi$ satisfies SA, then
$O_i^b(R_{-i})=O_i(R_{-i})$ for every $b\ge1$, since every attainable
object can be obtained from its singleton report, which remains
admissible. In particular, $O_i^1(R_{-i})=O_i(R_{-i})$.
Thus, every positive list-length constraint preserves
the unrestricted option set. Therefore, every positive list-length-constrained 
version of $\varphi$ remains SA.
However, unlike SA, strategyproofness need not survive a constraint on reports.
Haeringer and Klijn (2009) show that for DA and TTC, listing only
one's first $b$ acceptable objects in their true order need not be a
dominant strategy. A list-length constraint can therefore preserve
every attainable object while creating a strategic choice over which
objects to list.

While every positive list-length constraint preserves the unrestricted
option set for an SA mechanism, the same does not hold for an individually
rational non-SA mechanism. A one-object bound removes an attainable
object from some agent's option set in some market. If the mechanism has
unbounded width, no fixed list-length limit preserves all option sets
across markets: for every bound, there is a market in which some
attainable object requires a longer report.

Report width also governs best responses and equilibrium verification.
A report is \emph{$b$-bounded} if it declares at most $b$ objects
acceptable.

\begin{proposition}[Bounded best responses and equilibrium verification]\label{prop:best-response}
Fix $b\ge1$. For an individually rational mechanism, the following statements are equivalent:
\begin{enumerate}[label=(\roman*),nosep]
\item every attainable object can be attained by a $b$-bounded report;
\item for every agent $i$, every true preference $R_i$, and every
profile $R_{-i}$ of the other agents' reports, there exists a
$b$-bounded best response;
\item a submitted report profile is a Nash equilibrium if and only if no agent has a profitable $b$-bounded deviation.
\end{enumerate}
\end{proposition}

\begin{proof}
See Appendix~\ref{app:bounded-reports}.
\end{proof}

For an individually rational mechanism, Proposition~\ref{prop:best-response}
and the definition of $k_m(\varphi)$ imply that $k_m(\varphi)$ is
the smallest positive deviation-length bound sufficient to verify
Nash equilibrium across the markets in $\mathcal E_m$. Thus, unbounded
report width means that no fixed deviation-length bound suffices for
equilibrium verification as the market grows.\footnote{Proposition~\ref{prop:best-response} and this conclusion remain valid under the weaker condition that the empty
report assigns the outside option:
$\varphi_i(\emptyset,R_{-i})=\emptyset$ for every $i$ and $R_{-i}$.
This is the only consequence of individual rationality used in the
proof.}

\section{What causes singleton-attainability to fail?}\label{sec:failure}

The non-SA primary mechanisms and families lead to three broad sources
of indirect access. We first establish a general incompatibility
between SA and Pareto improvement over a
strategyproof, individually rational, and non-wasteful mechanism. We then introduce Positional
DA-TC mechanisms, a broad class of Pareto-efficient mechanisms that
Pareto-improve DA, and show that they have unbounded width. 
We also study two further structural classes which highlight
structural contrasts with their SA counterparts.
Stable Reverse Priority mechanisms contrast with Stable Monotone
Priority mechanisms, and Reported Rank Welfare mechanisms contrast
with Fixed Rank Welfare mechanisms. 
We show that both of these non-SA structural classes have maximal width.

\subsection{Pareto improvements over strategyproof mechanisms}\label{sec:da-width}

It is well known that DA is not Pareto efficient, but recent work
also demonstrates that its welfare losses can be considerable. In large
random markets, with high probability there is a Pareto improvement
over the DA outcome that makes almost all agents better off
(Ortega et al., 2026). Substantial gains in expected average
assignment ranks have also been established for Pareto-efficient
mechanisms that Pareto-improve DA (Ortega, Zhao, and Ziegler, 2026).

EADA and DA-TTC are two prominent mechanisms that Pareto-improve DA
and are Pareto efficient. We first illustrate that neither is
singleton-attainable.

\begin{example}[EADA and DA-TTC violate TRI and SA]\label{ex:da-improvements-ti}
Let $I=\{i,j,h\}$ and $C=\{c,d\}$ with unit capacities. The reports are
$R_i=(c,d)$, $R_j=(c)$, and $R_h=(d,c)$, and the priorities are
$\succ_c=(h,j,i)$ and $\succ_d=(i,h,j)$. DA assigns $(i\asgn d, j\asgn\emptyset, 
h\asgn c)$.

For EADA, agent $j$ is an interrupter at $c$. Deleting $c$ from her
list, so that she submits the empty report, and rerunning DA gives
$(i\asgn c, j\asgn\emptyset, h\asgn d)$. For DA-TTC, the DA assignments are the initial endowments. Since $i$ prefers $c$ to her endowment $d$ and $h$ prefers $d$ to her endowment $c$, the two agents point to one another and trade their endowed objects. This gives the same matching: $(i\asgn c, j\asgn\emptyset, h\asgn d)$.

If agent $i$ submits the singleton report $(c)$ instead of $(c,d)$,
DA assigns $(i\asgn \emptyset, j\asgn c, h\asgn d)$,  leaving $i$ unmatched.
This matching is already Pareto efficient, so both EADA and DA-TTC
leave it unchanged. Thus, listing $d$ below
$c$ triggers the DA rejection chain that enables $i$ to obtain $c$
through the subsequent Pareto improvement, whereas truncating below $c$ changes $i$'s assignment from $c$ to $\emptyset$ in both mechanisms. Hence, EADA and DA-TTC
violate TRI and SA. \hfill$\diamond$
\end{example}

The example reflects a more general impossibility which  does
not depend on stability or on the particular structure of DA. It uses only strategyproofness, together with individual
rationality and non-wastefulness, of the mechanism being
Pareto-improved.

\begin{theorem}[Pareto improvements over strategyproof mechanisms]\label{thm:pareto-improvement-nonsa}
Let $\varphi$ be strategyproof, individually rational, and non-wasteful. Every Pareto improvement over $\varphi$ violates SA.
Consequently, each Pareto improvement over $\varphi$ has width greater
than one.
\end{theorem}

\begin{proof}
Let $\psi$ be a Pareto improvement over $\varphi$, and choose a profile
$R$ and agent $i$ for whom the improvement is strict. Let
$c=\psi_i(R)$
and $c\,P_i\,\varphi_i(R)$.
Since $\varphi$ is individually rational, $c\in A_i(R_i)$.
Strategyproofness of $\varphi$ gives
$\varphi_i(R)\,R_i\,\varphi_i((c),R_{-i})$.
Since $c\,P_i\,\varphi_i(R)$, the mechanism $\varphi$ cannot assign
$c$ to $i$ when she reports $(c)$. Individual rationality therefore
implies
$\varphi_i((c),R_{-i})=\emptyset$.

Suppose, toward a contradiction, that $\psi$ satisfies SA. Since
$\psi_i(R)=c$, the object $c$ is attainable for $i$ at $R_{-i}$ under
$\psi$. Hence
$\psi_i((c),R_{-i})=c$.
Let
$ R'=((c),R_{-i})$,
$ \, \mu=\varphi(R')$,
$\nu=\psi(R')$.
Then $\mu_i=\emptyset$ and $\nu_i=c$. Since $\psi$ Pareto-improves
$\varphi$, every agent weakly prefers her assignment in $\nu$ to her
assignment in $\mu$ at $R'$. Every agent matched in $\mu$ strictly
prefers her assignment there to the outside option, given that $\mu$
is individually rational. Hence no agent matched in $\mu$ can become
unmatched in $\nu$, while agent $i$ becomes matched. Therefore, $\nu$
matches strictly more agents than $\mu$.

It follows that, for some object $d$,
$|\nu(d)|>|\mu(d)|$.
Since $\nu$ is feasible, $|\mu(d)|<q_d$. Choose an agent
$j\in\nu(d)\setminus\mu(d)$. Pareto improvement implies that
$\nu_j\,R'_j\,\mu_j$, where $\nu_j = d$. 
Since $d\neq\mu_j$, $\,d\,P'_j\,\mu_j$. Thus agent $j$ strictly prefers $d$ to her
assignment in $\mu$, while $d$ has an unassigned copy in $\mu$. This
contradicts the non-wastefulness of $\mu$.
Therefore, $\psi$ violates SA. 

Finally, a Pareto improvement over an
individually rational mechanism is individually rational.
Observation~\ref{obs:width-one} therefore implies that $\psi$ has
width greater than one.
\end{proof}

The contradiction in the proof of the theorem is derived through a participation argument.
Alva and Manjunath (2019b, Corollary 2) show that an individually rational and strategyproof mechanism admits no strategyproof Pareto improvement if its outcomes are Pareto-constrained participation-maximal: no allocation can make every agent weakly better off while strictly expanding the set of participants. 
Non-wastefulness implies this condition in our model (Alva and Manjunath, 2019b, Remark 4). Their assumptions of richness of the outside option and no indifference with the outside option hold in our domain.
Theorem~\ref{thm:pareto-improvement-nonsa} rules out all singleton-attainable Pareto improvements and strengthens their
result in the present matching environment: a strategyproof, individually
rational, and non-wasteful mechanism admits no
Pareto improvement within the larger class of SA mechanisms
not just within the class of strategyproof mechanisms.

The object allocation framework of Alva and Manjunath (2019b, Section 5.1) allows both contracts, which specify an object and its assignment terms, and general feasibility constraints. Our proof extends to this setting when agents can
list any single contract as their only acceptable assignment,
just as they can list a single object in our model.
This requirement is satisfied on the unrestricted domain of
strict preferences.\footnote{Their concept of
non-wastefulness, namely that no Pareto improvement can assign more agents
to some object, coincides with our definition of non-wastefulness in the
school choice model.}

Theorem~\ref{thm:pareto-improvement-nonsa} identifies an
incompatibility between singleton-attainability and Pareto improvement over a strategyproof mechanism with basic additional properties. It demonstrates a three-way tension among Pareto efficiency, transparent access, and preservation of individual welfare guarantees of a baseline strategyproof mechanism.  This is not a general conflict between Pareto efficiency and SA: TTC and its generalizations are Pareto efficient and singleton-attainable. The conflict arises from the individual welfare guarantees. 
In fact, the improving mechanism need not even be Pareto efficient. 

The corollary stated below is an important implication of this general incompatibility for the case where the guarantees are the DA assignments.

\begin{corollary}[Pareto improvements over DA]
\label{cor:da-improvement-nonsa}
Every Pareto improvement over DA violates SA. Consequently, each
Pareto improvement over DA has width greater than one.
\end{corollary}

\begin{proof}
DA is strategyproof (Dubins and Freedman, 1981; Roth, 1982a) and stable (Gale and Shapley, 1962). Stability implies individual rationality and non-wastefulness, so the result follows from Theorem~\ref{thm:pareto-improvement-nonsa}.
\end{proof}

Thus, a Pareto improvement over DA cannot preserve even  SA.  Corollary~\ref{cor:da-improvement-nonsa} implies the impossibility established by
Abdulkadiro\u{g}lu, Pathak, and Roth (2009, Theorem~1) with strategyproofness, instead of the weaker singleton-attainability requirement.

Relatedly, when reports are restricted to at most one acceptable object, whether
by constraint or voluntarily, DA is Pareto efficient. Every matched
agent receives her only acceptable object, so any Pareto improvement
would have to match an agent whom DA leaves unmatched. Such an agent
was rejected from her only acceptable object, which is therefore fully
assigned. Assigning it to her would displace an agent who would then
be worse off. DA therefore admits no Pareto improvement on this
restricted singleton preference domain.
This is consistent with the familiar understanding of the source of efficiency losses in the DA mechanism. Objects ranked above an agent's eventual assignment can create priority claims that affect other agents even though the reporting agent does not ultimately obtain those objects. 
Kojima and Manea's (2010) IR-monotonicity property formalizes this
effect: when agents withdraw some claims to acceptable objects ranked
above their respective original DA assignments, all agents are weakly
better off at the resulting DA outcome.

Kesten (2010) identifies an important procedural form of the same phenomenon: an agent can be held temporarily by an object she cannot ultimately obtain, causing another agent to be rejected and initiating a rejection chain whose downstream effects generate an efficiency loss. EADA removes the claims of such interrupters for these specific objects in reverse order of the DA procedure, undoing the associated rejection cycles. On the singleton-report domain, every agent receives either her only acceptable object or the outside option. A rejected agent has no other acceptable object to which she can move, so no rejection cycles can be set off and no Pareto improvement over the DA matching is possible. Once richer reports are allowed, such claims can generate the efficiency losses that Pareto improvements over DA seek to recover. Corollary~\ref{cor:da-improvement-nonsa} shows that this recovery comes
with its own loss of access transparency:  every Pareto improvement over DA violates singleton-attainability.

Corollary~\ref{cor:da-improvement-nonsa} does not determine the
severity of indirect access, which depends on the improvement
procedure. In particular, a Pareto improvement over DA can have
bounded width greater than one. For example, let a mechanism
select $(i\asgn c, j\asgn\emptyset, h\asgn d)$ at the initial
report profile in the market of
Example~\ref{ex:da-improvements-ti}, and select the DA matching
at every other profile and in every other market.
This mechanism Pareto-improves DA. At the initial profile, 
agent $i$ obtains $c$ with the report $(c,d)$ but not with
the singleton report $(c)$, so the mechanism has width at least
two. 
Report width is at most two in this market because there are
only two objects. In every other market, the mechanism coincides
with DA, which is singleton-attainable. Hence, this mechanism
has report width exactly two.
Importantly, this mechanism is not Pareto efficient: if $j$
reports $(c,d)$ while the other reports remain unchanged,
it selects the original, Pareto-inefficient DA matching.  
The next theorem shows that every mechanism in a broad
Pareto-efficient subclass of Pareto improvements over DA has
unbounded width, including EADA and DA-TTC. Whether every
Pareto-efficient Pareto improvement over DA must have unbounded
width remains an open question.

We use \emph{DA-TC} (DA Trading Cycles) for the broader class of
Pareto-efficient mechanisms that start from the DA matching and
Pareto-improve it through trading cycles. We define next a positional
subclass whose members have unbounded report width. To formulate
this subclass, we first introduce the trading structure and the
positional rules it uses.

For the trading stage, treat each object copy assigned by DA as a
distinct endowment held by its assigned agent. Only agents matched
by DA and their endowed copies enter the trading stage. Agents
continue to rank objects rather than individual object copies.

\paragraph{Improvement graph and cyclic agent set.}
At any stage, let $\mu$ denote the current matching of remaining
endowed agents, with $\mu_i$ denoting the object currently held by
agent $i$. The \textit{improvement graph} $G(\mu)$ has the remaining
endowed agents as vertices and contains an edge $i\to j$ whenever
$\mu_j \,P_i\, \mu_i$.
We call such an edge a \emph{claim}: agent $i$ is the
\emph{claimant}, and agent $j$ is the \emph{holder} of the
claimed copy.
Thus every directed cycle is a Pareto-improving trade.
The \emph{cyclic agent set} $\mathcal C(\mu)$ is the set of agents
that belong to some directed cycle of $G(\mu)$.
Before another trade is selected, every agent outside $\mathcal C(\mu)$
is settled at her current object and removed together with her endowed
copy; we refer to this as \textit{settlement}. 
As shown by Lemma~\ref{lem:cyclic-settlement} in
Appendix~\ref{app:positional}, a settlement is without loss:
these are exactly the remaining endowed agents whose assignments do
not change in any Pareto-improving reallocation of the current
endowed copies among the remaining endowed agents.

\bigskip

After settlement, when agents outside the cyclic agent set have been removed, claims
can be selected from either the claimant side or the holder side.
A claimant orders her remaining strict improvements by preference,
while a holder's remaining claimants are ordered by the object's
priority. In either orientation, we consider rules that select a
position in the ordered set. The selected position may depend on
the number of alternatives, but the same position is selected
whenever that number is the same.

A \emph{position rule} is a function
$f:\mathbb{N}_{+}\to\mathbb{N}_{+}$ satisfying
$1\leq f(s)\leq s$ for every $s\geq 1$.
When an ordered set contains $s$ alternatives, $f(s)$ specifies
the position selected from that set. The top and bottom rules
are $f(s)=1$ and $f(s)=s$, respectively.

\medskip

\defhead{Structural Class: Positional DA-TC}
A mechanism $\varphi$ is a \textit{Positional DA-TC mechanism} if it
uses one common position rule $f$ and one of the two orientations
below, both fixed independently of the report profile and used
throughout the trading stage. At every profile $R$, $\varphi(R)$ is
the final matching produced by the following two stages.

\medskip

\noindent \emph{Stage 1: DA endowments.}
Run DA and use its assigned object copies as the initial endowments.

\smallskip

\noindent \emph{Stage 2: Positional trading.}
At the current matching, first settle and remove the agents outside
the cyclic agent set and their current endowed copies (settlement).
If no agents remain, the trading stage ends.
Otherwise, apply the fixed orientation as follows:

\medskip

\begin{enumerate}[nosep]
\item \emph{Claimant-positional:}
Each remaining agent orders the distinct remaining objects that she
strictly prefers to her current endowment according to her report.
If there are $s$ such objects, she selects her $f(s)$-th
most-preferred one and points to the highest-priority remaining
holder of a copy of that object.

\item \emph{Holder-positional:}
For each remaining holder, order the remaining claimants for her
endowed copy of object $c$ according to $\succ_c$.
If there are $s$ such claimants, retain the claim of the
$f(s)$-th highest-priority claimant for $c$.
\end{enumerate}

The function $f$ is evaluated separately for each remaining agent
or holder, using the number of alternatives in that particular
ordered set. In the claimant-positional orientation, $s$ counts
the distinct objects that are held by remaining agents and that
the claimant strictly prefers to her current endowment. In the holder-positional
orientation, $s$ counts the remaining agents who strictly prefer
the holder's endowed object to their own current endowment.

\medskip

\noindent Since every remaining agent belongs to the cyclic agent
set, each remaining agent strictly prefers at least one remaining
object to her current endowment, and each remaining holder has at
least one remaining claimant for her endowed copy.
In the claimant-positional orientation, every remaining agent
therefore selects exactly one outgoing claim. In the holder-positional
orientation, every remaining holder retains exactly one incoming
claim. Thus, the selected claims form at least one directed cycle in each round, 
and these cycles are disjoint. 
Carry out all such
cycles simultaneously, with each participant receiving the endowed
copy held by the next agent along the directed cycle. Treat the
received copies as the new endowments. Agents who do not participate
in a selected cycle keep their current endowments. 

Update the improvement graph and the cyclic agent set at the
new matching after each trading round, and repeat these trading rounds as part of Stage~2 as long as any agent remains. Agents leave Stage~2 only when they are settled, so they may participate in a trade in  multiple consecutive trading rounds.  

\bigskip

The structural class extends substantially beyond the primary mechanisms
EADA and DA-TTC, and includes every positional claimant rule and
every positional holder rule.
For example, the claimant-positional bottom rule $f(s)=s$ has each remaining agent
select her least-preferred remaining strict improvement.
The top holder rule $f(s)=1$ is outcome-equivalent to full-consent
EADA through the Top Priority representation of Dur, Gitmez, and
Y{\i}lmaz (2019), while the top claimant rule $f(s)=1$ is
outcome-equivalent to DA-TTC. These equivalences are established
in the proof of Corollary~\ref{cor:eada-dattc-width} in
Appendix~\ref{app:positional}.
This class also allows for intermediate positions. For example, fixing
an integer $r\geq 1$ and setting $f(s)=\min\{r,s\}$ selects the
$r$-th alternative whenever at least $r$ alternatives are available,
and the last available alternative otherwise.

\begin{theorem}[Positional DA-TC mechanisms]\label{thm:positional-width}
Every Positional DA-TC mechanism is Pareto efficient, Pareto-improves
DA, and has unbounded report width.
\end{theorem}

\begin{proof}
See Appendix~\ref{app:positional}.
\end{proof}

The theorem identifies a procedural source of severe indirect access within the broader class of Pareto improvements over DA. Positional selection rules out an unrestricted global choice among Pareto-improving cycles: claims are selected locally from a claimant's
preference ordering or a holder's priority ordering. 
At the same time, the class remains broad enough to encompass
both priority-based and preference-based procedures, including
EADA and DA-TTC.

For the two primary mechanisms, the additional objects that matter are naturally in the tail, below the object ultimately obtained. These lower-ranked objects are not the assignments the agent seeks to obtain. Rather, they can trigger the rejection and trading steps that make the target object attainable, thereby allowing the agent to improve on her DA assignment.
The width construction shows how this logic can be repeated: as the market grows, more reported objects can become necessary to trigger successive rejection and trading steps, and omitting one can prevent an agent from obtaining the targeted object.

\begin{corollary}[Primary mechanisms in the Positional DA-TC class] \label{cor:eada-dattc-width}
EADA and DA-TTC have unbounded report width.
\end{corollary}

\begin{proof}
See Appendix~\ref{app:positional}.
\end{proof}

\subsection{Stable Reverse Priority mechanisms}\label{sec:srp}

Rank-Equity provides a prominent example of the structural class we
study next. Like Stable Monotone Priority mechanisms, members of this
class select a stable matching with respect to report-induced priorities
at every profile, but generate these priorities from reported
preferences in a very different way. Ekici, Ertemel, and Yenmez (2026)
propose an equity criterion for a one-to-one assignment problem with
complete strict rankings and obtain matchings satisfying it by selecting
stable matchings with respect to preference-induced priorities.
Agents who rank an object lower have higher priority for it, with
equal ranks resolved by an exogenous object-specific tie-breaker.
We extend this construction to our many-to-one model with incomplete
reports and apply DA using these induced priorities, as defined in
Appendix~\ref{app:mechanisms}.
Interestingly, Rank-Equity is the only primary mechanism
that violates TLI. As the following example illustrates, its failure
of SA arises from ranking the obtained object first, rather than from
deleting objects below it.

\begin{example}[Rank-Equity violates TLI and SA]\label{ex:rank-equity}
Let $I=\{1,2,3\}$ and $C=\{a,b,c\}$ with unit capacities, let all three agents report $(a,b,c)$, and use tie-breakers $\pi_a=(2,1,3)$ and $ \pi_b=\pi_c=(1,2,3)$.
At this profile the induced priority at each object is its tie-breaker, so Rank-Equity assigns $(2 \asgn a, 1 \asgn b, 3 \asgn c)$. 
If agent 1 lifts $b$ to the top, both other agents rank $b$ second while agent 1 ranks it first. Hence both have strictly higher induced priority than agent 1 at $b$. 
If a matching that is stable with respect to these induced priorities assigned $b$ to agent 1, each of the other two agents would have to receive $a$ to avoid blocking with $b$. This is not possible. Thus, Rank-Equity violates TLI. The same argument applies if agent $1$ submits the  singleton report $(b)$, proving a violation of SA. These implications follow because moving $b$ upward weakens the agent's claim to it. Furthermore, the argument is independent of the displayed tie-breakers. To see this, note that the same argument applies to any admissible tie-breaking profile when we let $i$ be the agent assigned $b$ at this common-report profile. \hfill$\diamond$
\end{example}

The example motivates the following condition.

\defhead{Reverse priority dominance}
A report-induced priority rule $\Gamma$ satisfies
\emph{reverse priority dominance} if, for every object $c$, every
report profile $R$, and every pair of distinct agents $i$ and $j$
with $c\in A_i(R_i)$, $c\in A_j(R_j)$,
\[
U_i(c;R_i)\subsetneq U_j(c;R_j)
\quad\ \mbox{implies} \quad
j\,\Gamma_c(R)\,i.
\]
Thus, when two agents' upper contour sets at $c$ are strictly nested,
the agent who lists at least one additional object above $c$ has
higher induced priority at $c$.

\medskip

The comparison with priority monotonicity becomes clear when one agent
removes some objects ranked above $c$. Suppose two agents have the
same nonempty head at $c$, with $i$ ranked above $j$ in the induced priority at $c$.
Now let only $i$ change her report so that her upper contour set at
$c$ becomes a proper subset of her original upper contour set, while
she continues to declare $c$ acceptable. Priority monotonicity
preserves $i$'s priority over $j$. Reverse priority dominance requires
the opposite priority comparison at the modified profile: $j$ must
be ranked above $i$ at $c$.
Thus, the two conditions point in opposite
directions, but they impose different types of restrictions.
Priority monotonicity only
preserves the reporting agent's priority over agents who were previously ranked 
below her, while reverse priority dominance determines the priority
comparison whenever one agent's upper contour set strictly contains
the other's.

A stronger cardinality-based requirement, closer to the priority
criterion used by Rank-Equity, would require that whenever $|U_i(c;R_i)|<|U_j(c;R_j)|$, we have  $j\,\Gamma_c(R)\,i$. This would also determine the priority comparison when the two upper
contour sets are not nested. Reverse priority dominance imposes the
comparison only under strict set inclusion and therefore leads to a
broader class of report-induced priority rules.

\defhead{Structural Class: Stable Reverse Priority}
A mechanism $\varphi$ is a \textit{Stable Reverse Priority mechanism}
if there is a report-induced priority rule $\Gamma$ satisfying reverse
priority dominance such that, for every profile $R$, $\varphi(R)$ is
stable at $R$ with respect to the priority profile $\Gamma(R)$.

\medskip

As with Stable Monotone Priority mechanisms, no restriction is imposed
on how $\varphi$ selects among stable matchings at $R$ with respect
to $\Gamma(R)$.

The contrast between the priority rules of the two classes is not
simply a reversal of direction. Merely replacing priority monotonicity
by the weak opposite condition, namely that reducing the head cannot
strengthen the reporting agent's priority, would not be enough to
imply a violation of SA. Fixed priorities satisfy this condition,
while fixed-priority stable mechanisms are SA by
Proposition~\ref{prop:stable}. Reverse priority dominance directly
prescribes the priority comparison whenever one agent's upper
contour set strictly contains another's. 

\begin{theorem}[Stable Reverse Priority mechanisms]\label{thm:srp-width}
Every Stable Reverse Priority mechanism has maximal report width.
\end{theorem}

\begin{proof}
Let $\varphi$ be a Stable Reverse Priority mechanism with
report-induced priority rule $\Gamma$. The statement is immediate for
$m=1$. Fix $m\geq 2$. Take $m+1$ agents and $m$ unit-capacity objects
$a_1,\ldots,a_m$, and let $R$ be the profile at which every agent
submits the common complete report $(a_1,a_2,\ldots,a_m)$. Every
stable matching at $R$ with respect to $\Gamma(R)$ assigns all $m$
objects, since otherwise an unmatched agent and an unassigned
acceptable object would form a blocking pair. Let $i$ be the agent
assigned $a_m$ by $\varphi(R)$.

Consider any report $R_i'$ of length at most $m-1$, and write
$R'=(R_i',R_{-i})$. If $R_i'$  omits $a_m$, individual rationality rules out
assigning $a_m$ to $i$. If it contains $a_m$, then it ranks at most $m-2$
objects above $a_m$. Hence, for every $j\neq i$,
\[
U_i(a_m;R_i')\subsetneq\{a_1,\ldots,a_{m-1}\}=U_j(a_m;R_j).
\]
Reverse priority dominance gives
$j\,\Gamma_{a_m}(R')\,i$ for every $j\neq i$.

Suppose that a matching $\mu$ which is stable at $R'$ with respect to
$\Gamma(R')$ assigns $a_m$ to $i$. There are $m$ other agents but only
$m-1$ other objects, so some other agent $j$ is unmatched in $\mu$.
Since $R_j$ is the complete report, $a_m\in A_j(R_j)$, and
$j\,\Gamma_{a_m}(R')\,i$. Hence $(j,a_m)$ blocks $\mu$,
a contradiction. 
Therefore, no stable matching at $R'$ with respect to $\Gamma(R')$
assigns $a_m$ to $i$.
This holds for every report $R_i'$ of length at most $m-1$, whereas
$\varphi_i(R)=a_m$. Hence $k_i^\varphi(a_m;R_{-i})=m$. Report length is
at most $m$, so $k_m(\varphi)=m$.

Since unit capacities are a special case of arbitrary capacities, the construction establishes the same lower bound for the many-to-one model.
\end{proof}

The reason for maximal width is especially simple. Given reverse priority dominance, objects ranked above the assigned object can keep the agent from losing priority at it. In the proof, everyone initially reports the same complete ranking.
If the agent who receives the last object removes even one object
from above it, every other agent has higher induced priority than her
at this object. The agent therefore has to keep all of these objects in her report, and their number grows with the market. 
By contrast, for the singleton-attainable Stable Monotone Priority mechanisms, ranking the assigned object higher cannot weaken the agent's induced priority for it, so the target is reached by its singleton report, without any need to rank objects above it.

\begin{corollary}[Primary mechanism in the Stable Reverse Priority class]\label{cor:rank-equity-srp}
Rank-Equity has maximal report width. 
\end{corollary}

\begin{proof}
Let $\Gamma$ denote the report-induced priority rule of Rank-Equity.
At every profile $R$, Rank-Equity selects the agent-optimal stable
matching with respect to $\Gamma(R)$. Fix a profile $R$, an object
$c$, and distinct agents $i$ and $j$ who both declare $c$ acceptable.
If $U_i(c;R_i)\subsetneq U_j(c;R_j)$, then
$\rk_i(c;R_i)<\rk_j(c;R_j)$, so
$j\,\Gamma_c(R)\,i$. The object-specific tie-breaker is irrelevant
because the reported ranks differ. Thus, $\Gamma$ satisfies reverse
priority dominance, and Rank-Equity is a Stable Reverse Priority
mechanism. Theorem~\ref{thm:srp-width} therefore implies that it has
maximal report width.
\end{proof}

\subsection{Reported Rank Welfare mechanisms}\label{sec:reported-outside}

Another source of indirect access arises when rank-welfare comparisons use the reported position of the outside option. 
We retain the welfare requirements of
Section~\ref{sec:fixed-outside} but use the reported assignment-rank list. In an
individually rational matching $\mu$, agent $i$'s \emph{reported assignment rank}
$\rho_i^P(\mu_i;R_i)$ is $\rk_i(\mu_i;R_i)$ if she is matched and $|A_i(R_i)|+1$ if she
receives the outside option. Write $\rho^P(\mu;R)=(\rho_i^P(\mu_i;R_i))_{i\in I}$. The
only change from the fixed convention is that the outside option has rank $|A_i(R_i)|+1$
rather than $m+1$. This is the \textit{reported} convention for treating the outside option, and it is used by three families of primary mechanisms, Reported Rawlsian, Reported Additive Rank Scoring, and Reported Rank-Maximal mechanisms.

Reported Rawlsian mechanisms apply the Rawlsian criterion to the reported rank list $\rho^P$.
Reported Additive Rank Scoring mechanisms apply the additive rank scoring criterion to the reported rank list $\rho^P$, with $g:\{1,\ldots,m+1\}\to\mathbb R$ strictly increasing. Reported Rank-Maximal mechanisms apply the rank-maximal criterion to $\rho^P$. All three families use report-independent tie-breaks. We show next that these families do not satisfy TRI and SA.

\begin{example}[Reported Rawlsian, Reported Additive Rank Scoring, and Reported Rank-Maximal mechanisms violate TRI and SA]\label{ex:reported}
Take three agents and three unit-capacity objects $a,b,c$, and let all three agents report $(a,b,c)$. 

A Reported Rawlsian matching assigns all three objects: a complete assignment has worst reported rank three, whereas any matching leaving an agent unmatched has worst reported rank four. Let $i$ be the agent assigned $a$ by a Reported Rawlsian mechanism. If $i$ truncates below $a$ and reports only $(a)$, her outside option moves to rank two. Any matching that still gives $a$ to $i$ either assigns some other agent to $c$ or leaves an agent unmatched at rank four, so its worst reported rank is at least three. By contrast, leaving $i$ unmatched and assigning the other two agents to $a$ and $b$ has worst reported rank two. Hence no Reported Rawlsian matching at the truncated profile assigns $a$ to $i$.

At the same profile, let $i$ be the agent assigned $a$ by a Reported Additive Rank Scoring mechanism. At the complete profile the minimum total rank score is $g(1)+g(2)+g(3)$. If $i$ truncates to $(a)$, any matching that still assigns $a$ to $i$ has total rank score at least $g(1)+g(2)+g(3)$, whereas leaving $i$ unmatched and assigning the other two agents to $a$ and $b$ has total rank score $g(1)+2g(2)$, which is strictly lower because $g(2)<g(3)$. 
Thus, no score-minimizing matching at the truncated profile assigns $a$ to $i$. 

The same profile gives the result for Reported Rank-Maximal mechanisms. At the complete profile, a rank-maximal matching assigns all three objects. Let $i$ be the agent assigned $a$. If $i$ truncates to $(a)$, any matching that still assigns $a$ to $i$ has one rank-one assignment and at most one rank-two assignment. By contrast, leaving $i$ unmatched and assigning $a$ and $b$ to the other two agents gives one rank-one assignment and two rank-two assignments, since $i$'s outside option now has rank two. The latter rank-count list is lexicographically better. Hence no Reported Rank-Maximal mechanism assigns $a$ to $i$ at the truncated profile. 

For all three families, deleting the objects ranked below $a$ from the report changes the reported rank of the outside option and alters the agent's assignment. Hence, none satisfy TRI or SA. 
 \hfill$\diamond$
\end{example}

All three families in Example~\ref{ex:reported} point to the same source of the violation of singleton-attainability: the treatment of the outside option. Under the reported convention, deleting objects below the assignment moves the rank of the outside option upward, while adding such objects pushes it downward. Tail entries can therefore change the welfare comparison without changing the rank of the assignment itself.

Recall the definitions of a welfare ordering $\unrhd$ and rank
monotonicity, together with the strictness requirement, from
Section~\ref{sec:fixed-outside}. The definition of the next structural
class, Reported Rank Welfare, uses the same welfare requirements but
applies them to reported rather than fixed assignment ranks.
Accordingly, an individually rational matching $\mu$ is
\textit{welfare-maximal} at profile $R$ if,
for every individually rational
matching $\nu$ at $R$,
$\rho^P(\mu;R)\unrhd\rho^P(\nu;R)$. 
Fix a strict report-independent ordering $\sigma$ of feasible
matchings to break welfare ties.

\defhead{Structural Class: Reported Rank Welfare}
A mechanism $\varphi$ is a \textit{Reported Rank Welfare mechanism}
if there exists a rank-monotonic welfare ordering $\unrhd$ satisfying
the strictness requirement such
that, at every profile $R$, $\varphi(R)$ is the $\sigma$-highest
among the individually rational matchings $\mu$ whose reported
assignment-rank list $\rho^P(\mu;R)$ is maximal according to
$\unrhd$.

\medskip

Thus, at every profile, the selected matching is welfare-maximal when
welfare is evaluated using reported assignment ranks, and it is the
$\sigma$-highest among all welfare-maximal matchings.

\begin{theorem}[Reported Rank Welfare]\label{thm:reported-outside}
Every Reported Rank Welfare mechanism has maximal report width.
\end{theorem}

\begin{proof}
Let $\varphi$ be a Reported Rank Welfare mechanism with welfare
ordering $\unrhd$.
The statement is immediate for $m=1$. Fix $m\ge2$. Take $m+1$ agents and $m$ unit-capacity objects $a_1,\ldots,a_m$, and let $R$ be the profile at which every agent reports $(a_1,a_2,\ldots,a_m)$.

Every welfare-maximal matching at $R$ assigns all objects. Otherwise there is both an unmatched agent and an unassigned object. Assigning that object to the unmatched agent produces another individually rational matching and changes one reported rank from $m+1$ to at most $m$, leaving all other ranks unchanged. The new rank list rank-dominates the old one and has fewer coordinates equal to $m+1$, so $\unrhd$ ranks it strictly above the old one.

Let $i$ be the agent whom the mechanism assigns $a_1$ at $R$. Now let $R_i'$ have length $k\le m-1$, and write $R'=(R_i',R_{-i})$. If $R_i'$ omits $a_1$, individual rationality rules out assigning $a_1$ to $i$. Suppose instead that $a_1$ appears in position $p\le k$, and consider any individually rational matching $\mu$ at $R'$ that assigns $a_1$ to $i$. Since there are $m+1$ agents and only $m$ objects, some other agent $j$ is unmatched. Reassign $a_1$ from $i$ to $j$, leave $i$ unmatched, and leave every other assignment unchanged. Call the resulting matching $\nu$, which is individually
rational.

Only the reported ranks of $i$ and $j$ change. Agent $i$'s
report has length $k$, so leaving her unmatched gives her
reported rank $k+1$. Agent $j$ still reports
$(a_1,a_2,\ldots,a_m)$, so assigning her $a_1$ changes her
reported rank from $m+1$ to $1$. Thus, the pair of reported
ranks for $(i,j)$ changes from $(p,m+1)$ to $(k+1,1)$. After sorting, $(1,k+1)$ is componentwise weakly below $(p,m+1)$ 
because $1 \le p$ and $k + 1 \le m < m + 1$. 
Since all other reported ranks are unchanged, replacing the
sorted pair $(p,m+1)$ by $(1,k+1)$ weakly lowers every entry
of the full sorted reported rank list.
Thus, the reported rank list of $\nu$ rank-dominates that
of $\mu$.
It also has one fewer coordinate equal to $m+1$. Therefore, $\unrhd$ ranks the reported rank list of $\nu$ strictly above that of $\mu$.

Thus, no welfare-maximal matching at $R'$ assigns $a_1$ to $i$. This holds for every report $R_i'$ of length at most $m-1$, while the complete report of length $m$ assigns $a_1$ to $i$ at $R$. Hence $k_i^\varphi(a_1;R_{-i})=m$. Report length is at most $m$, so $k_m(\varphi)=m$. Since unit capacities are a special case of arbitrary capacities, the construction establishes the same lower bound for the many-to-one model.
\end{proof}

The contrast with the fixed class is especially revealing.
Fix a welfare ordering $\unrhd$ satisfying the requirements of both
rank-welfare classes and a strict report-independent ordering $\sigma$
of feasible matchings, and let $\varphi^F$ and $\varphi^P$ be the
corresponding fixed and reported mechanisms.
They select the same matching at every complete report
profile, since their rank lists coincide there. The only difference is the rank
assigned to the outside option in incomplete reports. By
Theorems~\ref{thm:fixed-outside} and~\ref{thm:reported-outside} and
Observation~\ref{obs:width-one}, for every $m\ge 1$,
$\, k_m(\varphi^F)=1$ and $k_m(\varphi^P)=m$. 
Thus, changing only the outside option convention changes report width from one to maximal in the two structural classes. The fixed convention keeps its rank at $m+1$, while the reported convention makes it depend on report length. With the reported convention, preserving access to an object may require listing all $m$ objects, resulting in maximal width.  With the fixed convention, shortening the tail leaves the outside option rank unchanged, eliminating this source of indirect access.

\begin{corollary}[Primary mechanisms and families in the Reported Rank Welfare class]\label{cor:reported-rank-welfare-applications}
Reported Rawlsian mechanisms, Reported Additive Rank Scoring mechanisms, and Reported Rank-Maximal mechanisms have maximal report width.
\end{corollary}

\begin{proof}
The proof of Corollary~\ref{cor:fixed-rank-welfare-applications}
verifies rank monotonicity and the strictness requirement for all
three welfare orderings. Their reported versions are therefore
Reported Rank Welfare mechanisms, and
Theorem~\ref{thm:reported-outside} implies that they have maximal report width. 
\end{proof}

\subsection{One-sided invariance}\label{sec:complementary-invariance}

Examples~\ref{ex:da-improvements-ti}--\ref{ex:reported} identify whether the failure of singleton-attainability arises from the head or the tail of the report. We now ask whether the same mechanisms remain invariant to changes on the other side of the assignment. The dependence turns out to be one-sided for all the non-SA primary mechanisms and families. In addition, TLI holds for every mechanism in the Reported Rank Welfare class.

\begin{proposition}[One-sided invariance]\label{prop:complementary-invariance}
EADA, DA-TTC, and every Reported Rank Welfare mechanism satisfy TLI. Rank-Equity
satisfies TRI.
\end{proposition}

\begin{proof}
See Appendix~\ref{app:complementary-invariance}.
\end{proof}

Together with the violations demonstrated in Examples~\ref{ex:da-improvements-ti}--\ref{ex:reported}, this proposition completes the head-tail description of the primary mechanisms and families classified in Table~\ref{tab:class}.  
EADA, DA-TTC, Reported Rawlsian mechanisms, Reported Additive Rank Scoring mechanisms, and Reported Rank-Maximal mechanisms satisfy TLI but violate TRI: they are tail-dependent.\footnote{There is a parallel one-sided pattern for EADA in terms of incentives. Afacan et al.\ (2026) show that EADA is upper-manipulation-proof, while every Pareto improvement over DA must violate lower-manipulation-proofness. These are incentive properties, whereas TLI and TRI are assignment-invariance conditions.} Rank-Equity satisfies TRI but violates TLI: it is head-dependent. Put more concretely, lifting the assigned object does not change the assignment for the five tail-dependent mechanisms and families, while deleting everything below the assigned object does not change the assignment under Rank-Equity. Although one-sided invariance holds for all of the non-SA primary mechanisms and families classified in Table~\ref{tab:class}, mechanisms in general may violate both TLI and TRI.

\section{Structural classification and the width dichotomy}\label{sec:class-revisited}

\subsection{Structural classification}

The primary classification in Table~\ref{tab:class} can now be understood through the lens of the structural results presented in Sections~\ref{sec:causes-sa} and~\ref{sec:failure} (Theorems~\ref{thm:smp}, \ref{thm:fixed-outside}, \ref{thm:positional-width}, \ref{thm:srp-width}, and~\ref{thm:reported-outside}). 
We studied three pairs of contrasting structural classes.
Strategyproof mechanisms have width one, whereas Positional DA-TC
mechanisms have unbounded width. Stable Monotone Priority mechanisms have width one, whereas Stable Reverse Priority mechanisms have maximal width. Fixed Rank
Welfare mechanisms have width one, whereas Reported Rank Welfare
mechanisms have maximal width.
Table~\ref{tab:classes} groups the primary mechanisms and families
according to these six structural classes. 

The structural classes need not be disjoint, and the table does not
display every formal class membership. It groups the primary
mechanisms and families according to the structural explanations
developed in the paper. In particular, DA and Serial Dictatorship are
both individually rational and strategyproof and also belong to the
Stable Monotone Priority class.\footnote{Serial Dictatorship is a special case of DA when
every object uses the fixed agent order as its priority.
Its Stable Monotone Priority representation uses these common
priorities and does not satisfy stability with respect to an arbitrary priority profile.}

\begin{table}[H]
\centering
\scriptsize
\renewcommand{\arraystretch}{1.20}
\begin{tabularx}{\textwidth}{
>{\raggedright\arraybackslash}X
>{\centering\arraybackslash}p{0.14\textwidth}|
>{\raggedright\arraybackslash}X
>{\centering\arraybackslash}p{0.14\textwidth}}
\toprule
\multicolumn{2}{c|}{\textbf{Width one: SA}}&
\multicolumn{2}{c}{\textbf{Unbounded width}}\\
\cmidrule(lr){1-2}\cmidrule(lr){3-4}
\textit{Structural class} & \textit{Width} &
\textit{Structural class} & \textit{Width}\\
\hspace*{0.9em}\textit{and included primary} &  &
\hspace*{0.9em}\textit{and included primary} & \\
\hspace*{0.9em}\textit{mechanisms and families} &  &
\hspace*{0.9em}\textit{mechanisms and families} & \\
\midrule
\begin{minipage}[t]{\linewidth}\raggedright
\textbf{Strategyproof mechanisms}\\[0.18em]
\hspace*{0.9em}DA\\[0.23em]
\hspace*{0.9em}Serial Dictatorship\\[0.23em]
\hspace*{0.9em}TTC
\end{minipage}
 & \begin{minipage}[t]{\linewidth}\vspace{-0.70em}\centering
Width one: SA\\
(Proposition~\ref{prop:sp-inv},\\
Observation~\ref{obs:width-one})
\end{minipage}
&
\begin{minipage}[t]{\linewidth}\raggedright
\textbf{Positional DA-TC mechanisms}\\[0.18em]
\hspace*{0.9em}EADA\\[0.23em]
\hspace*{0.9em}DA-TTC
\end{minipage}
& \begin{minipage}[t]{\linewidth}\vspace{-0.70em}\centering Unbounded width\\
(Theorem~\ref{thm:positional-width},\\ Corollary~\ref{cor:eada-dattc-width})
\end{minipage}
\\
\midrule
\begin{minipage}[t]{\linewidth}\raggedright
\textbf{Stable Monotone Priority mechanisms}\\[0.18em]
\hspace*{0.9em}Object-Proposing DA\\[0.23em]
\hspace*{0.9em}IA\\[0.23em]
\hspace*{0.9em}PRP mechanisms\\[0.23em]
\hspace*{0.9em}Taiwan Deduction\\
\hspace*{1.8em}mechanisms
\end{minipage}
& \begin{minipage}[t]{\linewidth}\vspace{-0.70em}\centering Width one: SA\\(Theorem~\ref{thm:smp}, \\   Corollary~\ref{cor:smp-applications})\end{minipage}
&
\begin{minipage}[t]{\linewidth}\raggedright
\textbf{Stable Reverse Priority mechanisms}\\[0.18em]
\hspace*{0.9em}Rank-Equity
\end{minipage}
& \begin{minipage}[t]{\linewidth}\vspace{-0.70em}\centering Maximal width\\
(Theorem~\ref{thm:srp-width},\\ Corollary~\ref{cor:rank-equity-srp})
\end{minipage}
\\
\midrule
\begin{minipage}[t]{\linewidth}\raggedright
\textbf{Fixed Rank Welfare mechanisms}\\[0.18em]
\hspace*{0.9em}Fixed Rawlsian\\
\hspace*{1.8em}mechanisms\\[0.23em]
\hspace*{0.9em}Fixed Additive Rank\\
\hspace*{1.8em}Scoring mechanisms\\[0.23em]
\hspace*{0.9em}Fixed Rank-Maximal\\
\hspace*{1.8em}mechanisms
\end{minipage}
& \begin{minipage}[t]{\linewidth}\vspace{-0.70em}\centering Width one: SA\\(Theorem~\ref{thm:fixed-outside},\\ Corollary~\ref{cor:fixed-rank-welfare-applications})\end{minipage}
&
\begin{minipage}[t]{\linewidth}\raggedright
\textbf{Reported Rank Welfare mechanisms}\\[0.18em]
\hspace*{0.9em}Reported Rawlsian\\
\hspace*{1.8em}mechanisms\\[0.23em]
\hspace*{0.9em}Reported Additive Rank\\
\hspace*{1.8em}Scoring mechanisms\\[0.23em]
\hspace*{0.9em}Reported Rank-Maximal\\
\hspace*{1.8em}mechanisms
\end{minipage}
& \begin{minipage}[t]{\linewidth}\vspace{-0.70em}\centering Maximal width\\
(Theorem~\ref{thm:reported-outside},\\ Corollary~\ref{cor:reported-rank-welfare-applications})
\end{minipage}
\\
\bottomrule
\end{tabularx}
\caption{Structural classes and their width}
\label{tab:classes}
\end{table}

\subsection{Width dichotomy}\label{sec:width-dichotomy}

The preceding structural theorems reveal a distinct pattern at both
the primary mechanism and structural class levels, as summarized in
Tables~\ref{tab:class} and~\ref{tab:classes}: no intermediate bounded
report width case arises naturally among the broad range of mechanisms
and mechanism classes that we study.

\medskip
\begin{center}
\begin{minipage}{0.9\textwidth}
\textbf{Width dichotomy.} \emph{Every primary mechanism  and family classified in Table~\ref{tab:class} and every structural mechanism class classified in Table~\ref{tab:classes} has either width one or unbounded width.}
\end{minipage}
\end{center}
\medskip

\noindent 
The division is significant, but it is not a verdict against non-SA
mechanisms. For two of the non-SA sources we study, the loss of SA
is inherent in the objective. No SA mechanism can Pareto-improve a strategyproof, individually rational, and non-wasteful mechanism
(Theorem~\ref{thm:pareto-improvement-nonsa}), and no SA mechanism can
meet the reverse-rank equity criterion through stable selection
(Theorem~\ref{thm:srp-width}). Our results show that for the corresponding structural classes
in Table~\ref{tab:classes}, the gap is large: Positional DA-TC
mechanisms have unbounded width, while Stable Reverse Priority
mechanisms have maximal width.
The contrast between the Stable Monotone Priority and Stable Reverse
Priority classes is especially clear: in the former, moving the target
object upward cannot weaken the agent's priority for it, whereas in the
latter, ranking objects above the target may be needed to keep rivals from
moving ahead of her.

The dichotomy is more surprising for the two rank-welfare classes and the rank-welfare-based primary families. For these mechanisms the contrast is driven by a natural but contestable design
choice: whether the rank of the outside option is fixed or depends
on the submitted report.
For the primary families, making this rank report-dependent creates
tail dependence (Example~\ref{ex:reported} and
Proposition~\ref{prop:complementary-invariance}), while for the
entire Reported Rank Welfare class it generates maximal report width.
Therefore, this severe form of indirect access is not inherent in the underlying rank-welfare
objective: for each Reported Rank Welfare mechanism, keeping its
welfare ordering and tie-break but fixing the outside option rank
at $m+1$ gives a Fixed Rank Welfare mechanism, restoring
singleton-attainability and width one. However, since this change
also affects outcomes, the fixed convention need not be preferable.

Violating SA does not logically imply unbounded width. The artificial Pareto improvement over DA constructed in Section~\ref{sec:da-width} has report width exactly two, illustrating Theorem~\ref{thm:pareto-improvement-nonsa} which holds without requiring Pareto efficiency.  The example demonstrates this by modifying DA
only at one report profile in a two-object market.
In that market an agent can obtain an object by listing one additional
object, although she cannot obtain it with its singleton report.
In every other market, the mechanism coincides with DA and
provides singleton access to every attainable object. Thus,
bounded width greater than one can be engineered, for example, by confining the
departure from DA to a fixed small market. Similar bounded-width examples can also be found easily, such as variants of Serial Dictatorship which order agents according to the size of their report.

By contrast, the broad collection of primary mechanisms and the mechanisms in the structural classes are motivated by objectives and design principles
unrelated to report width, making it notable that none exhibits
intermediate bounded width.
Our structural results suggest why the dichotomy arises: the three unbounded-width results
(Theorems~\ref{thm:positional-width},~\ref{thm:srp-width},
and~\ref{thm:reported-outside}) have a common logic despite their different structural origins. In Positional DA-TC, additional reported objects can trigger further DA rejection chains, creating the trades that eventually bring the target to the agent. In Stable Reverse Priority mechanisms, objects ranked above the target can keep the agent from losing priority at it. In Reported Rank Welfare mechanisms, objects ranked below the target push the outside option farther down. These are very different reasons for listing additional objects, but what is common is that the same impact can be replicated through additional objects in the report: as the market grows, there may always be a role for more objects. This repeated role for additional objects results in unbounded and even maximal width for broad structural mechanism classes.

\section{Practical implications}\label{sec:implications}

The distinction between SA and non-SA mechanisms matters in practice. SA gives applicants direct access to attainable assignments, while with a non-SA mechanism access may depend on the broader report. Our analysis shows that these differences need not be marginal: in the broad non-SA classes we study, the reports required for access can grow arbitrarily long with the size of the market. This has practical implications for information requirements, disclosure, and robustness to reporting constraints. We also discuss how the resulting classification relates to the mechanisms used in practice.

\subsection{Information requirements}\label{sec:information}

Constructing a preference report is itself costly. Applicants and their families may need to learn which schools or programs exist, determine which ones are acceptable, form beliefs about admission prospects, and compare alternatives along multiple dimensions. Evidence from university and school admissions challenges the assumption that applicants enter the process with complete information even about their own preferences. Applicants may discover their preferences during the process, hold systematically inaccurate beliefs about admission chances, and make consequential reporting errors even if the mechanism is strategyproof.\footnote{See Grenet, He, and K\"ubler (2022) on preference discovery in university admissions, Arteaga et al.\ (2022) on beliefs and school search, Hakimov and Khanna (2025) on reporting errors when using strategyproof mechanisms, and Larroucau et al.\ (2025) on information frictions and application mistakes in centralized college admissions.} The length and content of a report therefore represent genuine informational demands rather than merely features of the formal message space.

Learning and ordering the alternatives one genuinely wishes to consider may be difficult under any mechanism. Under a non-SA mechanism, however, even an agent who knows her own preferences and has identified an attainable assignment may need to report other objects because of how they alter the course of the procedure or the criterion applied to the resulting matching. Identifying such entries therefore requires considering additional alternatives and, when their procedural role matters, understanding the mechanism and how one's report interacts with those of others.  Report width measures how large this procedural reporting requirement can become. Unbounded report width reflects substantial additional demand created by the procedure itself, since reporting additional objects may trigger rejections, interruptions, or endowment changes. For mechanisms with a maximal report width, all objects may have to be listed in the report, as each of the reported objects may change report-induced priorities, or alter the reported rank of the outside option. This additional reporting burden is absent from SA mechanisms.

The reporting burden may also be distributed unevenly. Information about schools, priorities, admission prospects, and procedural details may be easier to acquire for families with greater access to counselors, advice, or prior experience. Field interventions in school choice show that changing the amount and presentation of information can alter application and enrollment decisions, and that relatively simple decision supports may have larger effects than more extensive information provision.\footnote{See Corcoran et al.\ (2018) and Cohodes et al.\ (2025) on informational interventions in New York City high-school choice, and Hastings and Weinstein (2008) on simplified information about school quality.} A mechanism that requires applicants to understand indirect routes to an attainable school may therefore amplify differences in access to information and advice. Singleton-attainability removes this particular source of disparity.

\subsection{Disclosure and privacy}\label{sec:disclosure}

Disclosure is a related but distinct aspect of reporting. What an applicant must know to construct a report is one question; what the submitted report then reveals about her is another. A rank-order list may disclose which schools a family would be willing to accept and how it ranks them. To isolate disclosure from strategic reporting, suppose an applicant wishes to preserve the assignment she would receive with sincere reporting. Under an individually rational SA mechanism, once this school has been identified, the singleton report preserves the assignment without disclosing other acceptable schools and preferences over them. Singleton-attainability therefore permits the strongest possible form of report minimization for the individual. 
Observation~\ref{obs:replication} strengthens this to replication of the entire matching if the mechanism is nonbossy. 

When SA is not satisfied, preserving an assignment may require reporting additional schools and ordering information, including entries that matter because they alter the course of the procedure rather than because the applicant hopes to receive them. These entries nevertheless become part of the submitted record. Applicants may also differ substantially in how much reported preference information they must disclose in order to preserve their assignments. Report width measures minimum disclosure in terms of report length: the greatest number of schools that may be indispensable in a shortest report preserving access to an attainable school.

Submitted preference lists also matter for school districts and clearinghouses, which use them to estimate demand, assess welfare, plan capacity, and evaluate reforms (Abdulkadiro\u{g}lu, Agarwal, and Pathak, 2017; Agarwal and Somaini, 2018). Yet submitted rankings need not coincide with true preferences, especially when a mechanism rewards strategic reporting. For the Boston (IA) mechanism, for example, a stated first choice need not be the applicant's true first choice, as documented empirically by Abdulkadiro\u{g}lu et al. (2006) and Agarwal and Somaini (2018). Unbounded report width creates a related interpretive problem: some entries may be indispensable because of their procedural effects on access to a desired school, and yet the submitted lists are later used as preference data. SA therefore provides a data-minimization benchmark for both applicants and the institutions that use their reports.

\subsection{Reporting constraints and school choice reforms}\label{sec:reforms}

Report width clarifies one effect of limiting the number of schools that students may list. Such limits are already known to matter for incentives and welfare: they lead families to strategize about which schools to list, including whether to reserve space for safer options, and they can affect the stability and efficiency of the resulting assignment (Haeringer and Klijn, 2009; Calsamiglia, Haeringer, and Klijn, 2010; Decerf and Van der Linden, 2021). Report width captures a different effect of list-length constraints: whether they remove schools from the student's unrestricted option set.
As shown in Section~\ref{sec:width}, every constrained version of an SA mechanism remains SA, and every positive list-length constraint preserves the student's unrestricted option set, holding the other students' reports fixed. 
Without SA, a list-length constraint can prevent a student from
obtaining a school even when she includes it in her list.
The additional schools needed to obtain it may not all fit
within the permitted list length.

The width dichotomy makes this distinction particularly important.
For every non-SA primary mechanism and family in
Table~\ref{tab:class} and every non-SA structural class in
Table~\ref{tab:classes}, report width is unbounded. Thus, the
additional reporting needed to obtain an attainable school is
not confined to a few extra entries: the number of schools that
may have to be listed can grow without bound as the market grows.

The list-constrained school-choice reforms studied by Pathak and
S\"onmez (2013) provide a useful application of the SA distinction to changes in actual assignment systems, adding another dimension along which school-choice reforms can be evaluated. Following the prohibition of first-preference-first admissions in England, many local authorities replaced constrained versions of First-Preference-First, including Boston (IA) as its all-schools special case, with constrained DA, often retaining the same list-length constraint. Chicago's selective-high-school reform replaced a four-choice Boston mechanism with a four-choice Serial Dictatorship in 2009, and then increased the permitted list from four to six schools in 2010 while retaining Serial Dictatorship. Brighton and Hove replaced a three-choice Boston mechanism with three-choice DA, and Newcastle increased its DA from three to four choices by 2010. Pathak and S\"onmez (2013) compare these reforms through vulnerability to manipulation, Bonkoungou and Nesterov (2023) revisit them by counting manipulating students, and Bonkoungou and Nesterov (2025) evaluate them from the perspective of fairness.

First-Preference-First, DA, and Serial Dictatorship are all SA. Thus, the mechanism changes in these reforms did not create or eliminate the need to list additional schools for procedural reasons. Likewise, relaxing the list-length constraint under DA or Serial Dictatorship did not remove a procedural access barrier related to report width. Since constrained SA mechanisms remain singleton-attainable, these reforms stayed within the SA class: the mechanism or the list-length constraint changed, but the transparent access property did not. Had any of these school reforms replaced an SA mechanism with a non-SA one, the reform would also have introduced the indirect-access problem measured by report width and could have altered the information and disclosure requirements associated with access.

\subsection{Alignment with practice}\label{sec:practice}

The substantial structural gap in access transparency that we study theoretically also matters for how easily mechanisms can be understood. This may help explain why SA mechanisms are more readily adopted and sustained in practice, even when non-SA mechanisms serve objectives that SA mechanisms cannot meet. 
Among the mechanisms in operational use, past or present, that we were able to classify from the available documentation,\footnote{All cases considered here concern mechanisms that map report profiles consisting of ranked object lists to a matching,
the framework in which we study singleton-attainability.
Our discussion does not extend to procedures whose reports
or dynamics fall outside this framework, such as the couples
version of the residency match, in which reports rank pairs
of positions, or iterated multi-offer admissions systems.}
all but one belong to the SA mechanism classes we study.
This pattern also holds for constrained versions, since list-length constraints preserve singleton-attainability. The school-choice reforms discussed in Section 10.3 reinforce this pattern: although they altered the mechanism or the permitted list length, they remained within the SA class.

Stable mechanisms provide one large group: DA is used in New York City and Boston and in centralized admissions in a number of countries (Roth, 2008; Pathak, 2011), variants of Object-Proposing DA are used in England, Finland, and French high-school admissions (Cantillon, 2017), and stable score-limit or related cutoff systems are used in Hungary, Ireland, Spain, Turkey, and Chile (Bir\'o and Kiselgof, 2015; Rios et al., 2021). The pattern extends beyond stable mechanisms. Immediate Acceptance, Parallel mechanisms, Serial Dictatorship, and the Taiwan Deduction mechanisms have all been used in centralized assignment and are singleton-attainable (Pathak and S\"onmez, 2013; Chen and Kesten, 2017; Dur et al., 2022).

This pattern is not limited to current practice. 
First-Preference-First, a PRP mechanism, was used by many English local authorities before being banned in 2007 (Pathak and S\"onmez, 2013), while the earlier Pan-London mechanism described by Agarwal and Somaini (2018), which used DA with rank-dependent priority adjustments at some schools, is a Stable Monotone Priority mechanism. 
French Priority mechanisms, studied by Bonkoungou (2020),
were formerly used for French university admissions before being  replaced by a dynamic sequential system. New Haven's former
mechanism also belongs to the French Priority subfamily of PRP mechanisms
(Agarwal and Somaini, 2018; Kapor, Neilson, and Zimmerman, 2020).
It was replaced by list-length-constrained DA in 2019
(Arteaga et al., 2022). 
New Orleans used TTC, a strategyproof mechanism, in the first year of OneApp before switching to DA in 2013 (Abdulkadiro\u{g}lu et al., 2020). 
The replacement mechanisms were also SA in England, New Haven, and New Orleans, so these changes preserved direct access to attainable schools.

Among the non-SA mechanisms that we study, a modified version of EADA provides the only verified operational use we have found: the Flemish government's secondary-school admissions system uses student-proposing DA and has recently introduced an optional optimization in the form of a partial-consent EADA which protects statutory priorities (Vlaamse overheid, 2025; Commissie inzake Leerlingenrechten, 2025).\footnote{Such selective protection of priorities has also been studied theoretically in related restricted versions of EADA (Dur, Gitmez, and Y{\i}lmaz, 2019; Kitahara and Okumura, 2024), while Cerrone, Hermstr\"uwer, and Kesten (2024) present an experimental study motivated directly by the Flemish initiative.} The partial-consent design with school-specific lotteries described by
Havermans and Wouters (2020) has unbounded report width, since it
reduces to full-consent EADA when no statutory priorities apply
(Corollary~\ref{cor:eada-dattc-width}).
Interestingly, the Flemish design studies note that schools ranked below those that students aim for may serve as a ``springboard'' to a higher-ranked assignment and describe the optimization as technically complicated and less explainable (Wouters, Havermans, and Groenez, 2019; Havermans and Wouters, 2020). 
This springboard metaphor captures how lower-ranked schools
can enable access to a higher-ranked assignment. Our analysis formalizes this access restriction as EADA's tail dependence and failure of singleton-attainability.
A real-world implementation of EADA, even in a restricted form, is especially informative, as EADA has received substantial theoretical attention (Troyan, Delacr\'etaz, and Kloosterman, 2020; Ehlers and Morrill, 2020; Reny, 2022, among others). 
Thus, Flanders provides a natural test of whether one of the most prominent non-SA mechanisms can operate successfully in practice.

Overall, the practice we survey aligns closely with our classification, while the one documented exception highlights the indirect access phenomenon and its transparency cost.

\section{Conclusion}\label{sec:conclusion}

Our results expose a striking fact. All primary mechanisms and families in
our classification and all structural mechanism classes we study fall on
one of two sides, with nothing in between: either an agent can reach any attainable object by listing it as the only acceptable one, or the number of objects an agent may have to list to reach one grows without bound as the market grows. Singleton-attainability is the property that draws this line. It identifies a form of access transparency and reporting simplicity not captured by stability, efficiency, fairness, or manipulability, namely how much of the agent's preference must be reported to obtain an attainable object.

This paper provides an option-set interpretation of SA, a head-tail
decomposition through top-lift and truncation invariance, and
implications for manipulation, best responses, and equilibrium
verification. Strategyproofness and stability each imply SA, while
Object-Proposing DA, PRP mechanisms, Taiwan Deduction mechanisms,
and the three primary fixed rank-welfare families we study provide
further familiar examples of SA mechanisms, all of which also satisfy the two invariances. 
Among the non-SA primary mechanisms and
families, EADA, DA-TTC, and the three primary reported rank-welfare
families are tail-dependent, while Rank-Equity is head-dependent.
We also show that no Pareto improvement over a
strategyproof, individually rational, and non-wasteful mechanism is SA. Since DA has all
three properties, every Pareto improvement over DA violates SA.
Our structural results explain the sources of width one and of
unbounded or maximal width through three pairs of contrasting classes:
strategyproof versus Positional DA-TC mechanisms, Stable Monotone
versus Stable Reverse Priority mechanisms, and Fixed versus
Reported Rank Welfare mechanisms. 
On the whole, our analysis develops
a general theory of access transparency in
matching, identifying the structural sources of direct and indirect
access and quantifying the severity of indirect access through
report width.
In addition, we discuss practical implications: singleton-attainability 
minimizes the reporting requirement and potential disclosure costs,
and retains access to attainable objects even with list-length
constraints. We also note the strong alignment between SA mechanisms
and operational use.

Studying the three broad non-SA structural classes reveals a common
reason for the width dichotomy: in each class, once a source of indirect access is present, it can be replicated. If SA holds, no object other than the assigned one is needed in the report to obtain that assignment. In the non-SA structural classes we examine, an additional reported object may be required because of some feature of the mechanism: it may trigger another rejection chain needed to improve on the DA assignment, keep the agent from losing priority at her assigned object, or push the outside option another rank lower. 
These effects need not stop after a fixed number of objects in the report
and can be repeated on a larger scale as the market grows.
The lower-bound constructions establishing unbounded and maximal width exploit precisely this possibility. Thus, for these mechanisms, unbounded width is the repeated manifestation of the same structural feature that causes the initial failure of singleton-attainability.

These findings raise an interesting design question. Can one design a mechanism with bounded report width greater than one in which the bound follows from a substantive feature of the procedure, rather than from an ad hoc restriction imposed solely to produce bounded width? Such a mechanism would permit indirect access while limiting the number of
additional reported objects that may be needed to obtain an attainable
object, and could be proposed on its own merits. A particularly interesting question concerns Pareto-efficient mechanisms that Pareto-improve DA. Every such mechanism must violate SA, and the canonical representatives of this class, EADA and DA-TTC, have unbounded report width. However, it remains an open question whether bounded width greater than one is possible in this class and, if so, whether such a mechanism would be appealing. More generally, one may ask which procedural features, if any, permit indirect access to remain bounded while achieving objectives that existing SA mechanisms do not achieve.

Beyond report width, a distinct direction is order complexity:
how precisely must the objects in a successful report be ranked,
and how robust is access to ordering mistakes?
Singleton-attainability also applies beyond the school choice model.
The stability argument extends directly to one-to-one two-sided matching. 
More generally, SA applies naturally to allocation and matching
environments in which each agent receives at most one object or
contract, has an outside option, and can declare any single object
or contract as her only acceptable assignment. Such environments
can accommodate richer feasibility constraints, as in refugee
matching, where each family is assigned to one location but its
size affects the capacity required.

Other environments require adapting the notion of direct access itself. With random assignment, where an agent's assignment is a lottery rather than an object, one can ask whether the same attainable lottery remains accessible when the agent reports only the objects in its support. With multi-unit demand or complementarities, as with couples and siblings, the relevant outcome may be a bundle or joint assignment. Dynamic and sequential markets and more complex exchange environments, such as kidney exchange, raise further questions about how direct access should be defined when assignments depend on earlier events or linked exchanges. 
Investigating whether a similar pattern arises in these additional
settings may clarify whether the width dichotomy we identified
is specific to the matching environment studied here or reflects
a broader feature of matching and allocation procedures.

\section*{Appendix}

\appendix
\section{Definitions of mechanisms}\label{app:mechanisms}
This appendix collects definitions and references for primary mechanisms and families classified in Table~\ref{tab:class} that are central to the analysis.

\paragraph{Deferred Acceptance.}
Given the priority orderings of the objects, (Agent-Proposing) DA (Gale and Shapley, 1962) proceeds in rounds. Each unmatched agent applies to her highest-ranked acceptable object
that has not rejected her, if one remains. Each object considers its
new applicants together with the agents it is already tentatively
holding, tentatively holds the highest-priority agents up to its capacity,
and rejects the rest. The procedure ends when no rejection occurs. The resulting matching is the agent-optimal stable matching.

\paragraph{Object-Proposing DA.}
Object-Proposing DA (Gale and Shapley, 1962) is the version of DA in which objects apply to agents according to their priority orderings. 
In each round, each object with unfilled capacity applies, according to its priority
order, to agents to whom it has not previously applied, up to its
remaining capacity. Each agent considers the new applications together
with the object she is currently holding, tentatively holds her
most-preferred acceptable object among these if there is one, and
rejects the rest. The procedure ends when no rejection occurs. The resulting matching is the object-optimal stable matching.

\paragraph{Immediate Acceptance.}
IA uses the same application process as DA, but every acceptance is permanent. In each round, each unmatched agent applies to her highest-ranked acceptable object that has not rejected her, if one remains. Each object permanently accepts its highest-priority applicants up to its remaining capacity and rejects the rest. The procedure ends when no rejection occurs.

\paragraph{PRP mechanisms.}

A Preference Rank Partitioned (PRP) mechanism is specified by
report-independent partitions of priority ranks and preference ranks
into consecutive classes. The priority rank partitions may differ across objects, and the
preference rank partitions may differ across agents. At each report profile, each object orders agents who declare it acceptable first by priority rank class, then
within each such class by the preference rank class in which agents
place the object, and finally by its strict priority. In both
partitions, earlier rank classes have higher priority. The mechanism
applies DA with these report-induced priorities
(Ayoade and P\'apai, 2023).

\paragraph{Taiwan Deduction mechanisms.}
We use the following version of the Taiwan mechanism of
Dur et al. (2022), adapted to our incomplete-report framework
and allowing a report-independent strict tie-breaker for score
ties. A Taiwan Deduction mechanism fixes report-independent
base scores $s_{ic}$ for every agent-object pair, a weakly
increasing deduction schedule
$0=d_1\le d_2\le\cdots\le d_m$, and a report-independent
strict tie-breaker at each object.
 If agent $i$ lists object $c$ at rank $r$,
her adjusted score at $c$ is $s_{ic}-d_r$. Each object ranks the
agents who declare it acceptable in decreasing order of their
adjusted scores, breaking ties by its tie-breaker. A Taiwan Deduction
mechanism applies DA with these report-induced priorities.

\bigskip

\noindent\textit{Notation.} For the three rank-based criteria below, let $\rk_i(c;R_i)$ denote the position of an acceptable object $c$ in $R_i$. 
For a market with $m$ objects, a report $R_i$, and an individually rational matching $\mu$, define agent $i$'s \textit{fixed} and \textit{reported assignment ranks} as follows:
\[
\rho_i^F(\mu_i;R_i)=
\begin{cases}
\rk_i(\mu_i;R_i), & \mu_i\in A_i(R_i),\\
m+1, & \mu_i=\emptyset
\end{cases}
\qquad
\rho_i^P(\mu_i;R_i)=
\begin{cases}
\rk_i(\mu_i;R_i), & \mu_i\in A_i(R_i),\\
|A_i(R_i)|+1, & \mu_i=\emptyset
\end{cases}
\]
The corresponding \textit{fixed} and \textit{reported rank lists} are the following:
\[
\rho^F(\mu;R)=\bigl(\rho_i^F(\mu_i;R_i)\bigr)_{i\in I}
\qquad \qquad
\rho^P(\mu;R)=\bigl(\rho_i^P(\mu_i;R_i)\bigr)_{i\in I}
\]
For the next three pairs of mechanism families, optimization is over
the individually rational matchings at the reported profile.
Remaining ties are resolved by selecting the $\sigma$-highest
matching for a strict report-independent ordering $\sigma$
of feasible matchings.

\paragraph{Fixed and Reported Rawlsian mechanisms.}
Afacan and Dur (2024) study Rawlsian matching in an object-allocation
model with capacities and an outside option. Given a rank list
$\rho(\mu;R)=(\rho_i(\mu_i;R_i))_{i\in I}$, the Rawlsian criterion
first minimizes the worst assignment rank,
$\max_{i\in I}\rho_i(\mu_i;R_i)$,
and, subject to that, minimizes the number of agents receiving
that rank. Object priorities do not enter the criterion.
A \textit{Fixed Rawlsian mechanism} applies this criterion to the
fixed rank list $\rho^F$, so the outside option has rank $m+1$,
independently of the number of objects reported as acceptable.
A \textit{Reported Rawlsian mechanism} applies the same criterion
to the reported rank list $\rho^P$, so the outside option has rank
$|A_i(R_i)|+1$ for agent $i$. Given our report representation,
Afacan and Dur's assignment-rank list for individually rational
matchings coincides with $\rho^P$. They allow for selecting
arbitrarily among Rawlsian matchings. For the mechanisms considered
here, we resolve remaining ties using the report-independent ordering
$\sigma$ as specified above.

\paragraph{Fixed and Reported Additive Rank Scoring mechanisms.}
Given a rank list $(\rho_i(\mu_i;R_i))_{i\in I}$ and a fixed strictly
increasing rank scoring function $g:\{1,\ldots,m+1\}\to\mathbb R$,
the additive rank scoring criterion evaluates an individually
rational matching $\mu$ by 
\[
\sum_{i\in I} g\bigl(\rho_i(\mu_i;R_i)\bigr),
\]
with lower values preferred. A \textit{Fixed Additive Rank Scoring
mechanism} minimizes this sum using the fixed rank list $\rho^F$,
and a \textit{Reported Additive Rank Scoring mechanism} minimizes
the same sum using the reported rank list $\rho^P$.
The special case $g(r)=r$ gives the rank-sum criterion studied
by Ortega and Klein (2023).

\paragraph{Fixed and Reported Rank-Maximal mechanisms.}
We apply the rank-maximal criterion of Irving et al.\ (2006) to
the assignment-rank lists defined above. The criterion compares
matchings lexicographically by the numbers of agents receiving
successive assignment ranks, first maximizing the number receiving
rank~1, subject to that maximizing the number receiving rank~2,
and so on through rank~$m+1$. A \textit{Fixed Rank-Maximal mechanism}
applies this criterion to the fixed rank list $\rho^F$, and a
\textit{Reported Rank-Maximal mechanism} applies the same criterion
to the reported rank list $\rho^P$. Agents receiving the outside
option are included in these counts at rank $m+1$ in the fixed
version and at rank $|A_i(R_i)|+1$ in the reported version.

\paragraph{EADA.}
We study the full-consent EADA mechanism of Kesten (2010) in the many-to-one school-choice model. The mechanism first runs DA. An agent is an interrupter for an object if, while she is tentatively held by the object, another agent is rejected from it, and she is herself rejected by that object at a later step. The mechanism identifies the last step at which an interrupter is rejected from an object for which she is an interrupter. For every such interrupter-object pair at that step, it deletes the object from the agent's report, preserving the relative order of the remaining objects, and reruns DA with the updated reports. It repeats this until no interrupter remains, and selects the resulting DA matching. Tang and Yu (2014) give a simplified outcome-equivalent formulation, and Dur, Gitmez, and Y{\i}lmaz  (2019) provide an outcome-equivalent cycle formulation, the full-consent Top Priority mechanism.

\paragraph{DA-TTC.}
DA-TTC first runs DA and then applies TTC using the DA assignments
as initial endowments (Cantala and P\'apai, 2014; Alcalde and
Romero-Medina, 2017). Each agent matched by DA
is endowed with the object copy assigned to her. Agents left
unmatched by DA and object copies left unassigned by DA do not
enter the trading stage. In each TTC round, every remaining agent
identifies her most-preferred object among the remaining
endowments, including her own, and points to the highest-priority
remaining holder of a copy of that object, according to the
object's priority. All directed cycles are executed simultaneously,
with each participating agent receiving the endowed copy held by
the agent to whom she points. A self-cycle therefore leaves the
agent with her own endowment. The participating agents and their
endowed copies are then removed. The procedure repeats until no
agents remain in the trading stage.

\paragraph{Rank-Equity.}
Ekici, Ertemel, and Yenmez (2026) apply DA to
preference-induced priorities in a one-to-one
model with equally many agents and objects and no outside option. 
We extend their construction
to the many-to-one model with incomplete reports and call
the mechanism \emph{Rank-Equity}.
Fix a report-independent strict tie-breaker $\pi_c$ at each
object $c$. At every report profile $R$, the report-induced
priority $\Gamma_c(R)$ orders the agents who declare $c$
acceptable by the rank position at which they place $c$, with
lower-ranked positions receiving higher priority and equal ranks
broken by $\pi_c$. Agents who do not declare $c$ acceptable are ineligible
at $c$ and are placed last, in the order $\pi_c$.
Rank-Equity applies DA to the priorities $\Gamma(R)$ at each profile $R$.

\section{Omitted proofs}\label{app:proofs}

This appendix contains the proofs omitted from the main text.

\subsection{Extremal stable selection}
\label{app:proof-extremal-selection}

We first state two implications of independence and priority
monotonicity.

\begin{lemma}[Induced-priority changes]
\label{lem:induced-priority-changes}
Let $\Gamma$ satisfy independence and priority monotonicity. Fix an
agent $i$, an object $x$, reports $R_i$ and $R_i'$, and a profile
$R_{-i}$ such that $x \in A_i(R_i)$ and $x \in A_i(R'_i)$.
\begin{enumerate}[label=(\alph*),nosep]
\item If $U_i(x;R_i')=U_i(x;R_i)$, then
$\Gamma_x(R_i',R_{-i})=\Gamma_x(R_i,R_{-i})$.
\item If $U_i(x;R_i')\subseteq U_i(x;R_i)$, then every priority
comparison between two agents other than $i$ is unchanged, and
$i$ can only move weakly upward in $\Gamma_x(R_i',R_{-i})$ compared to $\Gamma_x(R_i,R_{-i})$.
\end{enumerate}
\end{lemma}

\begin{proof}
Independence preserves every priority comparison between two agents other
than $i$. If $U_i(x;R_i')=U_i(x;R_i)$, every comparison involving $i$ is also unchanged, so 
$\Gamma_x(R_i',R_{-i})=\Gamma_x(R_i,R_{-i})$, proving part~(a). If $U_i(x;R_i')\subseteq U_i(x;R_i)$, then priority monotonicity implies that for all $j \neq i$,  
if $i\, \Gamma_x(R_i,R_{-i}) j$ then $i\, \Gamma_x(R'_i,R_{-i}) j$, proving part~(b). 
\end{proof}

For two matchings $\mu$ and $\nu$ that assign the same number of
agents to an object $x$, say that $\mu(x)$ \emph{priority-dominates}
$\nu(x)$ with respect to a given priority order if, after the two
sets are listed in priority order, every position in $\mu(x)$ is
occupied by a weakly higher-priority agent than the corresponding
position in $\nu(x)$. At a fixed preference and priority profile,
the object-optimal stable matching priority-dominates every other
stable matching at each object and gives every agent her worst stable
assignment (Roth and Sotomayor, 1990).

\begin{proof}[Proof of Proposition~\ref{prop:extremal-selection}]
Fix a Stable Monotone Priority mechanism $\varphi$ with
report-induced priority rule $\Gamma$, an agent $i$, and a profile
$R=(R_i,R_{-i})$. Let $\mu=\varphi(R)$ and $\mu_i=c$. Since $\mu$
is stable at $R$, the object $c$ is acceptable to $i$. Write
$R_i=(H,c,T)$, where $H$ and $T$ denote the possibly empty ordered
lists of acceptable objects above and below $c$. For simplicity,
we also write $H$ and $T$ for the corresponding sets.

\bigskip

\noindent TLI: 
Let $R'=(R_i^{\uparrow c},R_{-i})$, where
$R_i^{\uparrow c}=(c,H,T)$. We first show that $\mu$ remains stable
at $R'$ with respect to $\Gamma(R')$. No blocking pair can involve
$i$, since she is assigned her first choice. At $c$, her upper
contour set shrinks to the empty set, so
Lemma~\ref{lem:induced-priority-changes}(b) implies that her priority
can only improve, while every comparison among the other agents
remains unchanged. At every other object, $i$ is not assigned, so
changes in her priority cannot create a blocking pair involving
another agent, while independence preserves all comparisons among
the other agents. Hence $\mu$ remains stable at $R'$.

\smallskip

\noindent \emph{Agent-optimal selection.}
Suppose first that $\varphi$ selects the agent-optimal stable
matching. Since the stable matching $\mu$ assigns $i$ her first
choice at $R'$, the agent-optimal stable matching at $R'$ also
assigns her $c$. Thus $\varphi_i(R')=c$.

\smallskip

\noindent \emph{Object-optimal selection.}
Suppose now that $\varphi$ selects the object-optimal stable matching,
and let $\nu=\varphi(R')$. Suppose, toward a contradiction, that
$\nu_i\neq c$.

We first show that $\nu$ is stable at $R$. Returning from
$R_i^{\uparrow c}$ to $R_i$ cannot create a new blocking pair
involving an agent other than $i$. At $c$, $i$'s priority can only
move weakly downward, and $i$ is not assigned there. At an object in $H$, her
priority can only improve. At every object in $T$, her upper contour
set is unchanged. Together with independence, these observations
rule out any new blocking pair involving another agent. A new
blocking pair involving $i$ can therefore only be of the form
$(i,d)$ for some $d\in H$.

Suppose that $(i,d)$ blocks $\nu$ at $R$. Since $\nu_i\neq c$, the
comparison of $d$ with $\nu_i$ is the same at $R_i$ and
$R_i^{\uparrow c}$. Thus $i$ also prefers $d$ to $\nu_i$ at $R'$.
Stability of $\nu$ at $R'$ implies that $d$ is fully assigned.
The blocking pair at $R$ therefore provides an agent
$j\in\nu(d)$ such that
$i\,\Gamma_d(R)\,j$.

Since $\mu_i=c$ and $d\,P_i\,c$, stability of $\mu$ at $R$ implies
that $d$ is fully assigned in $\mu$ and every member of $\mu(d)$ has
higher priority than $i$ at $\Gamma_d(R)$. Hence every member of
$\mu(d)$ has higher priority than $j$ at $\Gamma_d(R)$. Independence
preserves these comparisons at $R'$. Since both $\mu(d)$ and
$\nu(d)$ fill the capacity of $d$, $\nu(d)$ cannot priority-dominate
$\mu(d)$ at $\Gamma_d(R')$. This contradicts object-optimality of
$\nu$ at $R'$. Therefore, $\nu$ is stable at $R$.

Object-optimality of $\mu$ at $R$ now implies that $\mu(c)$
priority-dominates $\nu(c)$ at $\Gamma_c(R)$. Moving from $R$ to
$R'$ can only improve $i$'s priority at $c$, while independence
preserves every comparison among the other agents. Since
$i\in\mu(c)\setminus\nu(c)$, $\mu(c)$ also priority-dominates
$\nu(c)$ at $\Gamma_c(R')$. Object-optimality of $\nu$ at $R'$
gives the reverse priority dominance. The rural hospitals theorem
implies that the two sets have the same size, and two equally sized
sets can priority-dominate each other in a strict priority order only
if they are identical. This contradicts
$i\in\mu(c)\setminus\nu(c)$. Hence $\nu_i=c$.

\smallskip
Thus, TLI holds for either extremal stable selection.

\bigskip

\noindent TRI:
Let $R''=(R_i^{\le c},R_{-i})$, where
$R_i^{\le c}=(H,c)$. We first show that $\mu$ remains stable at
$R''$ with respect to $\Gamma(R'')$. Agent $i$ remains assigned
$c$ and deletes only objects below it, so no new blocking pair can
involve her. At every object in $H\cup\{c\}$, her upper contour set
is unchanged, and Lemma~\ref{lem:induced-priority-changes}(a)
implies that the entire induced priority order is unchanged. At
every other object, $i$ is not assigned, while independence
preserves every priority comparison among the other agents. Hence
no new blocking pair involving another agent can arise.

\smallskip

\noindent \emph{Agent-optimal selection.}
Suppose that $\varphi$ selects the agent-optimal stable matching.
In the agent-proposing DA procedure that produces $\mu$ at $R$,
agent $i$ never applies to an object in $T$, since her final
assignment is $c$. Truncation leaves the induced priority order
unchanged at every object to which she applies, while independence
preserves every comparison among the other agents. Thus the same
proposals, tentative holdings, and rejections occur at $R''$ as at
$R$. Hence $\mu$ is also the agent-optimal stable matching at
$R''$, and $\varphi_i(R'')=c$.

\smallskip

\noindent \emph{Object-optimal selection.}
Suppose that $\varphi$ selects the object-optimal stable matching.
Since $\mu$ is stable at $R''$ and matches $i$, the rural hospitals
theorem implies that every stable matching at $R''$ matches her.
The object $c$ is her last-ranked acceptable object at $R_i^{\le c}$ and
is assigned to her in the stable matching $\mu$. It is therefore
her worst stable assignment at $R''$. Hence the object-optimal
stable matching at $R''$ also assigns $c$ to $i$.

\smallskip

\noindent Thus, TRI holds for either extremal stable selection. 
\end{proof}

\subsection{Bounded reports, best responses, and equilibrium verification}\label{app:bounded-reports}

\begin{proof}[Proof of Proposition~\ref{prop:best-response}]

\leavevmode

\emph{(i) implies (ii).}
Fix $i$, $R_i$, and $R_{-i}$, and let $x$ be the best element of $O_i(R_{-i})\cup\{\emptyset\}$ according to $R_i$. If $x\in C$, condition~(i) supplies a $b$-bounded report attaining it. If $x=\emptyset$, the empty report attains it by individual rationality. Thus a $b$-bounded best response exists.

\emph{(ii) implies (iii).}
If a submitted report profile is not a Nash equilibrium, some agent has a profitable deviation. Her best response is therefore also profitable, and by (ii) she has a $b$-bounded best response. Hence she has a profitable $b$-bounded deviation. The converse is immediate, since every profitable $b$-bounded deviation is a profitable deviation.

\emph{(iii) implies (i).} 
Suppose that \textup{(i)} is
violated for some $i$, $R_{-i}$, and attainable object $c$, and let
$R_i=(c)$ be agent $i$'s true preference. Consider the submitted
profile $(\emptyset,R_{-i})$. By individual rationality, agent $i$
is unmatched. Since $c$ is her only acceptable object and no
$b$-bounded report attains it, she has no profitable $b$-bounded
deviation. Since $c$ is attainable, however, some unrestricted
report attains it and is a profitable deviation.
For every other agent $j$, choose her true preference so that her
assignment at $(\emptyset,R_{-i})$ is her top choice. If she is
unmatched, this means declaring every object unacceptable. No other
agent then has a profitable deviation. Thus, the submitted profile
has no profitable $b$-bounded deviation but is not a Nash
equilibrium, contradicting \textup{(iii)}. Therefore,
\textup{(i)} holds.
\end{proof}

\subsection{Positional DA-TC mechanisms}  \label{app:positional}

We first justify the settlement used in the definition of Positional DA-TC mechanisms and in the proof of Theorem~\ref{thm:positional-width}.
We then prove Theorem~\ref{thm:positional-width} and Corollary~\ref{cor:eada-dattc-width}.

\begin{lemma}[Settlement]\label{lem:cyclic-settlement}
For an endowed agent $i$ at an individually rational matching $\mu$ of endowed copies, the following are equivalent:
\begin{enumerate}[label=(\roman*),nosep]
\item $i$ belongs to a directed cycle of the improvement graph $G(\mu)$;
\item there is a matching $\nu$, obtained by reallocating the endowed copies among the endowed agents, that is a Pareto improvement over $\mu$ and satisfies $\nu_i\neq\mu_i$.
\end{enumerate}
Consequently, an agent $j$ outside the cyclic agent set $\mathcal C(\mu)$ at $\mu$ does not participate in any Pareto-improving trade after any sequence of Pareto-improving trades starting from $\mu$. Every matching $\nu$ obtained by a Pareto-improving sequence satisfies $\nu_j=\mu_j$.
\end{lemma}

\begin{proof}

\leavevmode

\emph{(i) implies (ii).}
If $i$ belongs to a directed cycle of $G(\mu)$, carrying out this
cycle and leaving all other assignments unchanged gives a
Pareto-improving reallocation that changes $i$'s assignment.
Thus condition~(ii) holds.

\emph{(ii) implies (i).}
Suppose a matching $\nu$ satisfies condition~(ii). Then $\nu$
reallocates the same endowed copies among the same endowed
agents. The changed assignments decompose into disjoint
cycles. Along each cycle, every agent strictly prefers the
object she receives in $\nu$ to her endowment in $\mu$.
Since $\nu_i\neq\mu_i$ by condition~(ii), agent $i$ belongs
to one of these cycles. This is a directed cycle of $G(\mu)$
containing $i$, which implies~(i).

\smallskip

\noindent This proves the equivalence of (i) and (ii).

\medskip

\noindent For the final claim, let $j\notin\mathcal C(\mu)$ and consider
any matching $\nu$ obtained by a sequence of Pareto-improving
trades starting from $\mu$. Every agent weakly prefers her
assignment in $\nu$ to her assignment in $\mu$. If
$\nu_j\neq\mu_j$, condition~(ii) holds for $j$.
By the equivalence of (i) and (ii), $j\in\mathcal C(\mu)$,
a contradiction. Thus every such matching satisfies
$\nu_j=\mu_j$. Since each trade strictly improves every
participant's assignment, $j$ cannot participate in any of
these trades.
\end{proof}

\begin{proof}[Proof of Theorem~\ref{thm:positional-width}]

\leavevmode
\medskip

We first establish that the trading stage is well defined and produces
a Pareto-efficient matching that every agent weakly prefers to the DA
matching. We then construct a family of markets in which attaining a
particular object requires reports of increasing length. After
determining the DA outcome in these markets, we prove a lower bound
on report size that holds for both positional orientations and show
separately that the bound is attained in each orientation. Finally, the
construction establishes unbounded report width.

\noindent \emph{Welfare properties.}
Every executed cycle in $G(\mu)$ strictly improves the assignment of all of its
participants. After settlement, every remaining
agent belongs to the current cyclic agent set. Hence a
claimant-positional graph has one outgoing edge from every remaining
vertex, while a holder-positional graph has one incoming edge into
every remaining vertex. Therefore, in either orientation, a cycle
exists whenever agents remain, and the procedure terminates since the
market is finite. Every agent is weakly better off at the final
matching than at the DA matching.

Any copy left unassigned by DA remains unassigned throughout the
procedure. By non-wastefulness of the DA matching, no agent strictly
prefers the object corresponding to such a copy to her DA assignment.
Since every agent is weakly better off at the final matching, no agent
strictly prefers such an object to her final assignment either.

Consider a matching $\nu$ that Pareto-improves the final matching.
Every agent matched by DA receives an acceptable object there and is
weakly better off at the final matching, so she cannot become
unmatched in $\nu$. Agents unmatched by DA do not enter the trading
stage and remain unmatched at the final matching. If one of them
became matched in $\nu$, then $\nu$ would match strictly more agents
than the final matching. Since the final matching uses every copy
assigned by DA, this would require assigning some agent a copy left
unassigned by DA, which the preceding paragraph rules out. Any
remaining Pareto improvement would therefore only reallocate the
copies assigned by DA among the agents endowed with them. Repeated
application of Lemma~\ref{lem:cyclic-settlement} during settlement
rules this out. Hence the final matching is Pareto efficient. Since DA is not Pareto efficient in general, the Pareto improvement
is strict at some profile.

\medskip
\noindent\emph{Construction and DA outcome.}
Fix $\ell\ge1$, and write $x_{\ell+1}=w$. There are agents
$i,h,r,\ a_t,b_t,p_t$ $(t=1,\ldots,\ell)$ and $3\ell+2$
unit-capacity objects $c,w,\ g_t,x_t,z_t$ $(t=1,\ldots,\ell)$.
For each $t$, call the agents $a_t,b_t,p_t$, together with the
objects $g_t,x_t,z_t$, \emph{block $t$}. Agent $i$ is the designated
agent and $c$ is the target object.

We will show that $i$ can obtain $c$ only if every block is active,
where block $t$ is active if $i$ applies to $g_t$ in the DA procedure.
Since this requires $g_t$ to appear in $i$'s report, every report
attaining $c$ must list $g_1,\ldots,g_\ell$ as acceptable. We
construct reports of length $\ell+2$ that activate all blocks and
deliver $c$ in either orientation.
Leaving $i$'s report unspecified, let $R_h=(x_1,c)$, $R_r=(w)$,
and, for every $t$, $R_{a_t}=(g_t,x_t)$ and
$R_{b_t}=(x_t,z_t)$. Agent $p_t$ reports
$z_t,x_{t+1},g_t$ as her only acceptable objects and ranks both
$z_t$ and $x_{t+1}$ above $g_t$. We will specify the relative order
of these two objects later according to the orientation and position
rule. Thus block $t$ is linked to the next block through $x_{t+1}$,
with $w=x_{\ell+1}$ closing the final block.

Let the priorities be
$\succ_c=(h,i,\ldots)$,
$\succ_w=(i,r,p_\ell,\ldots)$,
$\succ_{g_t}=(p_t,i,a_t,\ldots)$,
$\succ_{z_t}=(b_t,p_t,i,\ldots)$,
$\succ_{x_1}=(a_1,b_1,h,i,\ldots)$, and, for
$t=2,\ldots,\ell$,
$\succ_{x_t}=(a_t,b_t,p_{t-1},i,\ldots)$.
With these reports and priorities, the final DA assignments within
block $t$ depend only on whether the block is active. The table below
shows the assignments when the block is inactive versus active:
$$
\begin{array}{c|ccc}
&g_t&x_t&z_t\\ \hline
\text{inactive}&a_t&b_t&p_t\\
\text{active}&p_t&a_t&b_t
\end{array}
$$

If block $t$ is inactive, DA assigns $a_t\asgn g_t$ and
$b_t\asgn x_t$. Agent $p_t$ cannot be assigned to $x_{t+1}$:
for $t<\ell$, $x_{t+1}$ is eventually held by either $a_{t+1}$
or $b_{t+1}$, both of whom have higher priority than $p_t$ at
$x_{t+1}$, while for $t=\ell$, $w$ is assigned to either $i$ or
$r$, both above $p_\ell$. Since $b_t\asgn x_t$, $p_t$ can obtain
$z_t$, yielding the inactive row in the table.

If block $t$ is active, $i$ displaces $a_t$ from $g_t$, $a_t$ then
displaces $b_t$ from $x_t$, and $b_t$ moves to $z_t$, where she has
higher priority than $p_t$. Agent $p_t$ therefore cannot be assigned
to either $z_t$ or $x_{t+1}$ and eventually reaches $g_t$, where she
displaces $i$. This gives the active row in the table. Thus, reversing
the relative order of $z_t$ and $x_{t+1}$ in $p_t$'s report does not
change the displayed DA assignments. We will use this freedom below
to adapt the trading-stage claims to the orientation and position
rule without changing the DA endowments.

\medskip
\noindent\emph{Orientation-independent lower bound.}
Regardless of which blocks are active, $h$ is rejected from $x_1$
and DA assigns $h\asgn c$. Suppose that blocks
$1,\ldots,t-1$ are active and block $t$ is the first inactive block.
Define
$S_t=\{h,b_t\}\cup\{a_s,b_s,p_s:1\leq s<t\}$.
At the DA matching, these agents hold exactly
$D_t=\{c,x_t\}\cup\{g_s,x_s,z_s:1\leq s<t\}$.

Every object that a member of $S_t$ strictly prefers to her DA
assignment lies in $D_t$. Agent $h$ strictly prefers only $x_1$ to
her DA assignment $c$, while $b_t$ holds her first choice $x_t$.
For $1\leq s<t$, agent $a_s$ strictly prefers only $g_s$ to her DA
assignment $x_s$, agent $b_s$ strictly prefers only $x_s$ to her DA
assignment $z_s$, and agent $p_s$ strictly prefers exactly $z_s$ and
$x_{s+1}$ to her DA assignment $g_s$. Hence, no member of $S_t$
strictly prefers an object held by an agent outside $S_t$ to her DA
assignment.

We show by induction that the objects in $D_t$ remain assigned to
agents in $S_t$ throughout the trading stage. This holds at the DA
matching. Suppose it holds before a trade. Every member of $S_t$ is
weakly better off than at DA, so any object she strictly prefers to
her current assignment was also strictly preferred to her DA
assignment. By the characterization of the agents' strict
improvements above, that object lies in $D_t$ and, by the induction
hypothesis, is held by another member of $S_t$. Thus, any
Pareto-improving cycle involving a member of $S_t$ consists entirely
of agents in $S_t$. Executing such cycles keeps all objects in $D_t$
with agents in $S_t$, completing the induction. Since $c\in D_t$
and $i\notin S_t$, agent $i$ can never receive $c$.

Every report by $i$ that attains $c$ must therefore activate every
block and hence must list $g_1,\ldots,g_\ell$. Since DA assigns $c$
to $h$, agent $i$ must receive $c$ in the trading stage. Individual
rationality therefore requires her report to list $c$.
When every block is active, DA assigns
$a_t\asgn x_t$, $b_t\asgn z_t$, and $p_t\asgn g_t$ for every $t$,
while $h$ holds $c$. These assignments exhaust every object except
$w$. Since an unmatched agent does not enter the trading stage,
agent $i$ must be matched by DA, so her DA assignment must be $w$.
Individual rationality of DA therefore requires her report to list
$w$. Hence every report by $i$ that attains $c$ contains
$c,w,g_1,\ldots,g_\ell$ and has length at least $\ell+2$.

\medskip
\noindent\emph{Claimant-positional attainment.}
Suppose first that $\varphi$ is claimant-positional, with position
rule $f$. For each $t$, let
$R_{p_t}=(x_{t+1},z_t,g_t)$ if $f(2)=1$, and
$R_{p_t}=(z_t,x_{t+1},g_t)$ if $f(2)=2$.
Order $c,g_1,\ldots,g_\ell$ in $i$'s report so that $c$ occupies
position $f(\ell+1)$, and place $w$ last. This report has length
$\ell+2$ and activates every block. DA therefore gives
\[
i\asgn w,\quad h\asgn c,\quad
a_t\asgn x_t,\quad b_t\asgn z_t,\quad p_t\asgn g_t
\quad (t=1,\ldots,\ell).
\tag{1}
\]
At this matching, the full improvement graph contains the cycle
\[
i\to h\to a_1\to p_1\to a_2\to p_2\to\cdots
\to a_\ell\to p_\ell\to i.
\tag{2}
\]
For each $t$, it also contains the cycle
$b_t\to a_t\to p_t\to b_t$. Hence, all endowed agents and their
endowments remain after settlement.

The claimant-positional rule selects every claim in (2). In her
report, agent $i$ ranks exactly $\ell+1$ objects above her DA
assignment $w$, namely $c,g_1,\ldots,g_\ell$. Since $c$ occupies
position $f(\ell+1)$ among these objects, the claimant-positional
rule makes her select $c$. Agent $h$ strictly prefers only $x_1$ to
her DA assignment $c$, and each $a_t$ strictly prefers only $g_t$
to her DA assignment $x_t$. Holding $g_t$, agent $p_t$ has exactly
two strict improvements, $z_t$ and $x_{t+1}$, and $R_{p_t}$ makes
her select $x_{t+1}$. Thus every claim in (2) is selected, so the
cycle is carried out. Since $h$ holds $c$ in (1), agent $i$
receives $c$.

After this cycle, every $a_t$ holds her first choice $g_t$, so each
is settled at that object. Consequently, all $g_t$ are removed before
$i$ can participate in another trade. Once they are removed, $i$,
now holding $c$, has no strict improvement among the remaining
endowments and is therefore settled at $c$, even if some of the
$g_t$ objects were ranked above $c$. Thus a report of length
$\ell+2$ attains $c$.

\medskip
\noindent\emph{Holder-positional attainment.}
At a given matching, call an agent a \emph{cyclic claimant} for an
object if she belongs to the cyclic agent set and strictly prefers
that object to her current assignment.

Suppose next that $\varphi$ is holder-positional, with position rule
$f$. Let $R_{p_t}=(z_t,x_{t+1},g_t)$ for
$t=1,\ldots,\ell$, and let
$R_i=(c,g_1,\ldots,g_\ell,w)$.
All blocks are active, so the DA matching is again (1). At this
matching, (2) is a cycle in the full improvement graph, and for each
$t$, $b_t\to a_t\to p_t\to b_t$ is another improvement cycle.
Hence all agents in these cycles belong to the cyclic agent set.

The endowed objects with two cyclic claimants have the following
claimants, listed from highest to lowest priority:
\[
\begin{array}{c|c}
\text{object}&\text{cyclic claimants in priority order}\\ \hline
g_t&i,\ a_t\\
x_1&b_1,\ h\\
x_t\ (t\geq2)&b_t,\ p_{t-1}
\end{array}
\]
Every other endowed object has only one cyclic claimant: only $i$
claims $c$, only $p_t$ claims $z_t$, and only $p_\ell$ claims $w$
among endowed agents, since agent $r$ is unmatched. Since trades only
improve assignments, every object an agent claims later is also
strictly preferred to her DA assignment. Thus, every remaining object
has at most two cyclic claimants throughout the trading stage.
Since $f(1)=1$, the only possibilities to consider are $f(2)=1$ and
$f(2)=2$.

Suppose first that $f(2)=2$. The retained claims at the objects in
the table are $a_t\to p_t$ at $g_t$, $h\to a_1$ at $x_1$, and
$p_{t-1}\to a_t$ at $x_t$ for $t\geq2$. The claims $i\to h$ at
$c$ and $p_\ell\to i$ at $w$ are also retained. Thus every claim
in (2) is retained, so that cycle is carried out. Since $h$ holds
$c$ in (1), agent $i$ receives $c$. The claims
$p_t\to b_t$ at $z_t$ are also retained, but they cannot belong to
another trading cycle because no claim made by $b_t$ is retained when
$f(2)=2$.

Suppose now that $f(2)=1$, so each remaining holder retains the claim
of her highest-priority cyclic claimant. The trading stage then works
backward through the blocks. The successive selected cycles are
$C_\ell=(i,p_\ell)$, then
$C_t=(i,p_t,a_{t+1})$ for $t=\ell-1,\ldots,1$, and finally
$C_0=(i,h,a_1)$.
Initially $i$ is the highest-priority cyclic claimant for $g_\ell$,
while $p_\ell$ is the only cyclic claimant for $w$, so $C_\ell$ is
carried out. No claim made by $h$ or any $a_t$ is retained, so no
other trading cycle arises at this stage. If $\ell=1$, this first
cycle is already $C_1$. For $1\leq t<\ell$, suppose inductively that
$C_\ell,\ldots,C_{t+1}$ have been carried out and the agents already
shown to be non-cyclic have been settled and removed. At the beginning
of this round, the assignments are
$i\asgn g_{t+1}$, $a_{t+1}\asgn x_{t+1}$, and
$p_t\asgn g_t$.

Before the next trade, $b_{t+1}$ and $p_{t+1}$ are non-cyclic
agents. Agent $b_{t+1}$ strictly prefers only $x_{t+1}$ to her
assignment, and this object is held by $a_{t+1}$, who strictly
prefers only $g_{t+1}$ to her assignment. The latter object is held
by $i$. Agent $i$ strictly prefers only
$c,g_1,\ldots,g_t$ to her assignment. Any improvement path from
$i$ stays among $i$, $h$, the agents in blocks $1,\ldots,t$, and
$a_{t+1}$. It therefore cannot reach $p_{t+1}$, the only agent who
claims $z_{t+1}$, and hence cannot return to $b_{t+1}$. Thus
$b_{t+1}$ is not in a cycle. Agent $p_{t+1}$ holds $x_{t+2}$ and
strictly prefers only $z_{t+1}$ to it. Since $z_{t+1}$ is held by
$b_{t+1}$, she is non-cyclic as well. Both are therefore settled and
removed.

The claims $i\to p_t$, $p_t\to a_{t+1}$, and
$a_{t+1}\to i$ are all retained: $i$ is the highest-priority
remaining cyclic claimant for $g_t$, $p_t$ is the highest-priority
remaining cyclic claimant for $x_{t+1}$, and $a_{t+1}$ is the only
cyclic claimant for $g_{t+1}$. No other trading cycle is carried out,
since $i$'s claim is retained at each $g_s$ with $1\leq s\leq t$,
so no claim made by these $a_s$ is retained, and $b_1$'s claim is
retained at $x_1$, so $h$'s claim is not retained. Hence $C_t$ is
carried out. It assigns $g_{t+1}$ to $a_{t+1}$, her first choice, so
she is settled. Repeating the argument gives $C_1$.

After $C_1$, the same argument makes $b_1$ and then $p_1$ non-cyclic.
Both are settled and removed. The claims $i\to h$, $h\to a_1$, and
$a_1\to i$ are all retained, so $C_0=(i,h,a_1)$ is carried out.
Since $h$ still holds $c$, agent $i$ receives $c$.

In both cases, $c$ is $i$'s first choice, so she is settled. Thus,
for either value of $f(2)$, the report of length $\ell+2$ attains
$c$.

\medskip

\noindent\emph{Width conclusion.}
For either orientation, the lower-bound and attainment arguments
give fixed reports $R_{-i}$ for which
$k_i^\varphi(c;R_{-i})=\ell+2$. Since the constructed market
has $3\ell+2$ objects,
\[
k_{3\ell+2}(\varphi)\ge \ell+2
\qquad\text{for every }\ell\ge1.
\]
Thus, for the constructed market sizes $m=3\ell+2$, we obtain
the linear lower bound $k_m(\varphi)\ge(m+4)/3$.
Therefore, $k(\varphi)=\infty$.
\end{proof}

\medskip

\noindent \textit{Remark: Position-rule extensions. }
The claimant-positional part of the theorem continues to hold if each agent $j$ has her own position rule $f_j$. The welfare argument is unchanged, and in the width construction one places $c$ in position $f_i(\ell+1)$ of $i$'s report and orders $z_t$ and $x_{t+1}$ in $p_t$'s report so that $p_t$ selects $x_{t+1}$ according to $f_{p_t}(2)$. The holder-positional width proof does not extend directly to object-specific position rules, because its attainment argument relies on the same value of $f(2)$ at every object.

\medskip

\begin{proof}[Proof of Corollary~\ref{cor:eada-dattc-width}]
We show that the holder-positional top rule $f(s)=1$ is
outcome-equivalent to EADA, while the claimant-positional top rule
$f(s)=1$ is outcome-equivalent to DA-TTC. Both mechanisms then belong
to the Positional DA-TC class, so the result follows from
Theorem~\ref{thm:positional-width}.

\medskip

\noindent\emph{EADA.}
Dur, Gitmez, and Y{\i}lmaz (2019) define the full-consent Top
Priority procedure, whose outcome coincides with the Pareto-efficient
full-consent EADA outcome. Starting from the DA matching, Top Priority
selects at each holder the claim of the highest-priority remaining
claimant for the holder's object, carries out one resulting cycle,
and continues. An agent settled by Top Priority cannot belong to a directed cycle
of strict improvement claims at that step or at any later step
(Dur, Gitmez, and Y{\i}lmaz, 2019, Section~4.2 and Lemma~8).
By comparison, the holder-positional top rule first settles every agent
outside $\mathcal{C}(\mu)$ and then selects, at each remaining
holder, the claim of the highest-priority cyclic claimant. We compare
the two procedures in two steps to show that the holder-positional top rule $f(s)=1$ is
outcome-equivalent to EADA.

\smallskip

\noindent\emph{(i) Every Top Priority cycle is selected by the
holder-positional top rule.}
Fix a current matching $\mu$ and a cycle carried out by Top Priority.
Every edge of this cycle is a strict improvement claim, so every
participant belongs to $\mathcal{C}(\mu)$. Moreover, every cyclic
agent is among the agents considered by Top Priority, since an agent
settled by Top Priority cannot belong to a directed cycle of
strict improvement claims. For each edge of the Top Priority cycle,
its claimant has highest priority among the claimants considered by
Top Priority and therefore also among the cyclic claimants. Hence
every edge of the cycle is selected by the holder-positional top
rule.

\smallskip

\noindent\emph{(ii) The order of selected cycles does not affect the
outcome.}
Suppose that $\alpha$ and $\beta$ are distinct cycles selected by
the holder-positional top rule at a matching $\mu$. The two cycles
are disjoint. After $\alpha$ is carried out, the assignments of the
agents in $\beta$ are unchanged, so its edges remain
strict improvement claims and its participants remain cyclic agents.
Consider an edge $i\to j$ of $\beta$. Suppose that after $\alpha$
a higher-priority cyclic claimant $i'$ is selected in place of $i$
at the object held by $j$. Since $i'$ is a cyclic agent after the trade,
Lemma~\ref{lem:cyclic-settlement} implies that she was already a cyclic agent
at $\mu$. Moreover, $\alpha$ weakly improves her assignment, so if
she still strictly prefers the object held by $j$, she also strictly
preferred it to her assignment at $\mu$. 
Agent $i'$ would therefore have
been selected ahead of $i$ before $\alpha$, a contradiction. 
Thus, $\beta$ will still be selected after $\alpha$ is carried out.
The same argument applies whenever two cycles are selected at the
same matching. Since the cycles are disjoint, carrying out both
in either order gives the same matching. Every trade strictly
improves its participants, and there are finitely many matchings,
so the sequential procedure terminates. These observations imply
that the order of trades does not affect the final matching.
In particular, carrying out selected cycles one at a time or
simultaneously gives the same outcome.

Now consider a sequential version of the holder-positional top rule
that carries out the Top Priority cycles in their original order.
Part~(i) shows that each such cycle is available when it is to be
carried out, while Lemma~\ref{lem:cyclic-settlement} ensures that
settlement cannot remove a participant in a later Top Priority cycle.
Following the Top Priority sequence of cycles therefore produces
the Top Priority matching, which is Pareto efficient. Hence no
improvement cycle remains at that matching. By part~(ii), carrying
out the selected cycles simultaneously produces the same matching.
Hence the holder-positional top rule is outcome-equivalent to EADA.

\medskip

\noindent\emph{DA-TTC.}
We show that the trading stage of the claimant-positional top rule is
outcome-equivalent to the TTC stage of DA-TTC. Fix a current
endowment matching $\mu$.

\smallskip

\noindent\emph{(i) Every nontrivial TTC cycle is selected by the
claimant-positional top rule.}
In TTC, an agent may point to a higher-priority holder of another
copy of her current object. Such an edge cannot be in a nontrivial
TTC cycle: while the higher-priority holder remains, no agent
selecting that object points to the lower-priority holder, so the
latter has no incoming edge. Every nontrivial TTC cycle therefore consists of strict improvement edges and is a directed cycle of $G(\mu)$. Hence, every agent in
such a cycle belongs to $\mathcal{C}(\mu)$. Thus, removing the agents outside
$\mathcal{C}(\mu)$ leaves this cycle unchanged.
Given that each agent in such a cycle points to the
highest-priority holder of her most-preferred remaining object, it is exactly
a claimant-positional cycle. Therefore, every nontrivial TTC cycle is
selected by the claimant-positional top rule.

Lemma~\ref{lem:cyclic-settlement} also implies that an agent outside
$\mathcal{C}(\mu)$ cannot participate in a later Pareto-improving
trade. Together with the preceding observation, such an agent keeps
her current object in TTC and eventually leaves through a self-cycle.

\smallskip

\noindent\emph{(ii) Every claimant-positional cycle is carried out
by TTC.}
Consider a cycle selected by the claimant-positional top rule after
the agents outside $\mathcal{C}(\mu)$ have been settled. Before any
member of this cycle leaves TTC, each member points either to the
next member of the cycle or to a holder outside
$\mathcal{C}(\mu)$. Consider the first TTC round in which a member
leaves. She cannot leave through a self-cycle, since she strictly
prefers the next agent's object in the cycle to her own. By part~(i), a
nontrivial TTC cycle cannot contain an agent outside
$\mathcal{C}(\mu)$. Hence, every agent in this TTC cycle follows
the claimant-positional cycle, and TTC carries out this cycle,
possibly in a later round.

After the claimant-positional rule carries out the cycle, no
participant strictly prefers any remaining object to the object she
has received, so each is settled at the next step. This corresponds to the removal of the participants and their endowed copies
in TTC. Other claimant-positional cycles selected at the same
matching are disjoint and are unchanged, so they can
be treated one at a time. Repeating the argument for the remaining
agents shows that the two procedures carry out the same nontrivial
cycles, possibly at different rounds, and produce the same final
matching. Hence the claimant-positional top rule is
outcome-equivalent to DA-TTC.
\end{proof}

\subsection{One-sided invariance}\label{app:complementary-invariance}

\begin{proof}[Proof of Proposition~\ref{prop:complementary-invariance}]
\leavevmode
\medskip

\noindent \emph{EADA.}
Kesten (2010, Lemma~A.4) shows that EADA satisfies positive
association. Given that it is also individually rational,
repeatedly interchanging the assigned object with the object
immediately above it moves it to the top without changing the
assignment. Hence EADA satisfies TLI.

\medskip
\noindent\emph{DA-TTC.}
Suppose DA-TTC assigns $c$ to $i$ at $(R_i,R_{-i})$. 
Let $w$ be agent $i$'s DA endowment at $(R_i,R_{-i})$.
Consider $(R_i^{\uparrow c},R_{-i})$. Since TTC weakly improves
every agent's assignment, either $c=w$ or $c\,P_i\,w$.
If $c=w$, strategyproofness of DA (Dubins and Freedman, 1981; Roth, 1982a) implies that moving $c$ to the top leaves her DA assignment unchanged. Since $c$ is now $i$'s
first choice, the TTC stage also assigns her $c$.
Suppose now that $c\,P_i\,w$. Moving $c$ to the top reorders
only objects above $w$, so strategyproofness of DA implies
that $i$'s DA assignment remains $w$. Since the upper contour
set of $w$ is unchanged, IR-monotonicity of DA implies that
the entire DA matching remains the same at
$(R_i^{\uparrow c},R_{-i})$
(Kojima and Manea, 2010, Theorem~1).
The endowments entering the TTC stage are therefore unchanged.
Since TTC with fixed endowments is strategyproof (Roth, 1982b),
if $i$ lifts her TTC assignment $c$ to the top she still obtains $c$. Hence, $i$ is assigned $c$ at $(R_i^{\uparrow c},R_{-i})$ and DA-TTC satisfies TLI.

\medskip

\noindent\emph{Reported Rank Welfare mechanisms.}
Suppose a Reported Rank Welfare mechanism assigns $c$ to $i$ at $(R_i,R_{-i})$.
Consider $(R_i^{\uparrow c},R_{-i})$.
The top lift preserves the set of acceptable objects and the length of
$i$'s report, so her reported outside option rank $|A_i(R_i)|+1$  is unchanged.
Agent $i$'s reported assignment rank weakly improves in the
originally selected matching at $(R_i,R_{-i})$.
In every matching that does not assign $c$ to her, her assignment rank is either unchanged or gets worse. All other agents' ranks and the set of individually rational matchings are unchanged. Thus, the TLI argument in the proof of
Theorem~\ref{thm:fixed-outside} applies with $\rho^P$ in place of $\rho^F$: rank
monotonicity and the report-independent tie-break rule out selecting a matching that
takes $c$ away from $i$. Hence, $i$ is assigned $c$ at $(R_i^{\uparrow c},R_{-i})$.  Thus, every Reported Rank Welfare mechanism satisfies TLI.

\medskip

\noindent\emph{Rank-Equity.}
Suppose Rank-Equity assigns $c$ to $i$ at $(R_i,R_{-i})$.
Consider $(R_i^{\le c}, R_{-i})$. 
Truncating below $c$ leaves the reported positions of $c$ and
every object above it the same, and hence leaves $i$'s induced
priority unchanged at every object to which she applies.
In the DA procedure at $(R_i,R_{-i})$, agent $i$ never applies to an object
ranked below $c$, which is also true at $(R_i^{\le c}, R_{-i})$. The relative priority
of all other agents is unchanged at every object.
Therefore, induction on the DA rounds gives the same proposals,
tentative holdings, and rejections at $(R_i^{\le c},R_{-i})$
and at $(R_i,R_{-i})$.
Hence, $i$ is assigned $c$ at $(R_i^{\le c}, R_{-i})$. Thus, Rank-Equity
satisfies TRI.
\end{proof}

\section{Obvious manipulation and singleton-attainability}\label{app:om}

Here we demonstrate that singleton-attainability is distinct from obvious manipulation, a notion of strategic transparency introduced by Troyan and Morrill (2020).

\medskip

 For a true preference $R_i$ and an alternative report $\widehat R_i$, let
\[
B_i(\widehat R_i;R_i)=\max\{\varphi_i(\widehat R_i,\widehat R_{-i}):\widehat R_{-i}\in\mathcal R_{-i}\},
\qquad
W_i(\widehat R_i;R_i)=\min\{\varphi_i(\widehat R_i,\widehat R_{-i}):\widehat R_{-i}\in\mathcal R_{-i}\}
\]
denote the best and worst assignments that the report can produce, evaluated by $R_i$.

\defhead{Obvious manipulation} A report $\widehat R_i$ is an \emph{obvious manipulation} at $R_i$ if it is a profitable manipulation for some $\widehat R_{-i}$ and either
\[
B_i(\widehat R_i;R_i)\,P_i\,B_i(R_i;R_i)
\quad\text{or}\quad
W_i(\widehat R_i;R_i)\,P_i\,W_i(R_i;R_i).
\]
A mechanism is \emph{obviously manipulable} (OM) if such a preference and report exist. A mechanism is \emph{not obviously manipulable} (NOM) if no obvious manipulation exists.

For any individually rational and non-wasteful mechanism, truthful
reporting attains the agent's best possible assignment for some reports
of the other agents. If agent $i$ has an acceptable object, let no other agent
declare her first-ranked acceptable object acceptable. Individual
rationality and non-wastefulness then imply that this object is assigned
to $i$. If agent $i$ has no acceptable object, individual rationality
ensures that truthful reporting always gives her the outside
option, her best possible assignment.
Hence, only the worst-case comparison can generate an obvious
manipulation. Moreover, a report whose worst-case outcome is strictly
better than that of truthful reporting is profitable at a profile where
truthful reporting attains its worst-case outcome. It follows that the
mechanism is NOM if and only if truthful reporting is a maximin strategy
for every agent and preference.\footnote{This equivalence is noted by
Chaudhury and P\'apai (2025) and Coreno and Balbuzanov (2026).}

The OM criterion of Troyan and Morrill (2020) concerns transparency of the incentive to deviate. SA concerns transparency of access to an object. 
For an individually rational SA mechanism, once the reports of the
other agents are fixed, the existence of a profitable manipulation
can be determined by testing singleton reports for objects preferred
to the truthful assignment. For the individually rational and
non-wasteful mechanisms compared here, OM checks whether a report
improves the agent's worst-case outcome across all possible reports
of the other agents.
Hence SA simplifies manipulation conditional on the environment, whereas OM identifies manipulations that are attractive without such information. The two concepts therefore capture different forms of strategic transparency: SA is ex post, profile-specific transparency, while in our setting OM is ex ante, worst-case transparency. The comparison provided in Table~\ref{tab:nom} shows that the two forms of strategic transparency may agree or differ for different mechanisms.

\begin{table}[htbp]
\centering
\renewcommand{\arraystretch}{1.22}
\begin{tabular}{>{\raggedright\arraybackslash}p{0.20\textwidth}|>{\raggedright\arraybackslash}p{0.31\textwidth}|>{\raggedright\arraybackslash}p{0.31\textwidth}}
\toprule
& \textbf{SA} & \textbf{Non-SA}\\
\midrule
\textbf{OM} & IA & Rank-Equity\\
\midrule
\textbf{NOM} & Object-Proposing DA & EADA, DA-TTC\\
\bottomrule
\end{tabular}
\caption{Prominent examples of OM/NOM and SA/non-SA mechanisms}
\label{tab:nom}
\end{table}

Troyan and Morrill (2020, Proposition 1) prove that IA is OM, and their Theorem 2 implies that Object-Proposing DA, EADA, and DA-TTC are NOM. 
Moreover, Rank-Equity is OM with the tie-breakers in
Example~\ref{ex:rank-equity}. Take agent 1's true preference to be
$R_1=(b,a,c)$ and consider the alternative report
$\widehat R_1=(a,b,c)$. This report always gives her either $a$ or $b$.
Indeed, rejection from $b$ would require an opponent who ranks
$b$ third, since agent 1 ranks it second and wins its tie-break.
That opponent must already have been rejected from both $a$ and
$c$, so all three objects would then be holding agents other
than agent 1, which is impossible. By contrast, truthful reporting
gives her $c$ when $R_2=(b,a,c)$ and $R_3=(a,b,c)$.
The deviation therefore strictly improves her worst-case
assignment and is an obvious manipulation.
Thus non-SA and OM are both established with the same agent, although the two properties quantify over profiles differently.

The diagonal cells in Table~\ref{tab:nom} show cases in which the two classifications agree. IA is both OM and SA, while EADA and DA-TTC are NOM and non-SA. The off-diagonal cells show that the two dimensions are independent. Object-Proposing DA is an example of a non-strategyproof mechanism which is NOM but SA, while Rank-Equity is OM (for some tie-breakers) but non-SA. The two forms of transparency, obvious manipulability and singleton-attainability, are therefore independent and need not share a common source.

There is also an interesting distinction between the two concepts within the set of stable-dominating mechanisms, as defined by Alva and Manjunath (2019a), mechanisms which select either a stable matching or a Pareto improvement over a stable matching at every profile. 
Troyan and Morrill (2020, Theorem 2) show that every stable-dominating mechanism is NOM. SA  distinguishes within this class: stable mechanisms are SA by Proposition~\ref{prop:stable}, whereas Corollary~\ref{cor:da-improvement-nonsa} implies that EADA and DA-TTC are not SA.

\end{document}